\documentclass{article}

\usepackage[a4paper,margin=3cm]{geometry}
\usepackage{authblk}
\usepackage{mathdots}
\usepackage{enumitem}
\usepackage{mdframed}
\usepackage{longtable}
\usepackage{graphicx}
\usepackage{amsthm}
\usepackage{amssymb, amsmath, amsxtra}
\usepackage{stmaryrd}
\usepackage{bbm}
\usepackage{galois}
\usepackage{listings}
\usepackage{color}
\usepackage{hyperref}
\usepackage{times}
\usepackage[ruled,onelanguage,noline,noend,titlenotnumbered]{algorithm2e}

\SetAlFnt{\small}
\SetAlCapFnt{\small}
\SetAlCapNameFnt{\small}
\SetAlCapHSkip{0pt}
\IncMargin{-\parindent}

\usepackage{array}
\newcolumntype{L}[1]{>{\raggedright\let\newline\\\arraybackslash\hspace{0pt}}m{#1}}
\newcolumntype{C}[1]{>{\centering\let\newline\\\arraybackslash\hspace{0pt}}m{#1}}
\newcolumntype{R}[1]{>{\raggedleft\let\newline\\\arraybackslash\hspace{0pt}}m{#1}}

\usepackage{tikz}
\usetikzlibrary{arrows,backgrounds,shapes}
\newcommand{\comment}[1]{}

\usepackage{mathtools}
\DeclarePairedDelimiter{\ceil}{\lceil}{\rceil}

\newcommand{\cG}{\mathcal{G}}
\newcommand{\bO}{\mathbb{O}}
\newcommand{\bR}{\mathbb{R}}
\newcommand{\bZ}{\mathbb{Z}}
\newcommand{\bQ}{\mathbb{Q}}
\newcommand{\bS}{\mathbb{S}}

\def\grasse#1{{\llbracket #1 \rrbracket}}
\def\ok#1{\mbox{\raisebox{0ex}[1ex][1ex]{$#1$}}}
\def \tuple#1{\langle #1 \rangle}

\newcommand{\Lra}{\Leftrightarrow}
\newcommand{\Ra}{\Rightarrow}

\newcommand{\ra}{\rightarrow}

\newcommand{\fwedge}{f_{\scriptscriptstyle \!\wedge}}
\newcommand{\fvee}{f_{\scriptscriptstyle \!\vee}}
\newcommand{\wpwedge}{\wp^{\scriptscriptstyle \wedge}}
\newcommand{\wpvee}{\wp^{\scriptscriptstyle \vee}}

\newcommand{\gwedge}{g_{\scriptscriptstyle \!\wedge}}

\newcommand{\pem}{\preceq_{\mathit{EM}}}
\DeclareMathOperator{\Abs}{Abs}
\DeclareMathOperator{\lfp}{lfp}
\DeclareMathOperator{\gfp}{gfp}
\DeclareMathOperator{\uco}{uco}
\DeclareMathOperator{\lne}{leq}
\DeclareMathOperator{\gne}{geq}
\DeclareMathOperator{\id}{id}
\DeclareMathOperator{\argmax}{argmax}
\DeclareMathOperator{\sgn}{\mathit{sgn}}
\DeclareMathOperator{\Sem}{Sem}

\DeclareMathOperator{\cl}{cl}
\DeclareMathOperator{\EM}{\mathit{EM}}

\DeclareMathOperator{\SL}{SL}

\DeclareMathOperator{\Fix}{Fix}

\DeclareMathOperator{\Eq}{Eq}
\newcommand{\triangleRel}{\mathrel{\overset\triangle\Longleftrightarrow}}
\newcommand{\ud}{\triangleq}

\newtheorem{theorem}{Theorem}[section]
\newtheorem{lemma}[theorem]{Lemma}

\newtheorem{corollary}[theorem]{Corollary}
\newtheorem{example}[theorem]{Example} 
\newtheorem{definition}[theorem]{Definition}

\begin{document}

\title{Abstract Interpretation of Supermodular Games}
\author{\href{http://www.math.unipd.it/en/people/francesco.ranzato}{Francesco Ranzato}}
\affil{{\small Dipartimento di Matematica, University of Padova, Italy}}
\date{}

\maketitle

\begin{abstract}
Supermodular games find significant applications in a variety of models, especially in operations research and 
economic applications of noncooperative game theory, and 
feature pure strategy Nash equilibria characterized as fixed points of multivalued functions on complete lattices. Pure strategy Nash equilibria of supermodular games are here approximated by resorting to the theory of abstract interpretation, a well established and known framework used for designing static analyses of programming languages. This is obtained by
extending the theory of abstract interpretation in order to handle approximations of multivalued functions  and 
by providing some methods for abstracting supermodular games, in order to obtain approximate Nash equilibria which are shown to be correct 
within the abstract interpretation framework. 
\end{abstract}

\section{Introduction}\label{intro}

\paragraph{Motivations.} Games may have strategic
complementarities, which means, roughly speaking, 
that best responses of players have monotonic reactions, reflecting a complementarity relationship between own actions and rivals' actions. 
Games with strategic complementarities occur in a large array of
models, especially in operations research and 
economic applications of noncooperative game theory, a significant sample of them is described by Topkis' book \cite{topkis98}.
Pionereed by Topkis \cite{topkis78}, this class of games is formalized by supermodular games, where the payoff functions of each player 
have the lattice-theoretical 
properties of supermodularity and increasing differences. In a supermodular game, the strategy space of every player is partially ordered and 
is assumed to be a complete lattice, while the utility in playing a higher strategy increases when the opponents also play higher strategies. 
It turns out that pure strategy Nash equilibria of supermodular games exist and form a complete lattice w.r.t.\ the ordering relation 
of the strategy space, thus exhibiting
the least and greatest Nash equilibria. Furthermore, since the best response correspondence of a supermodular game satisfies a monotonicity
hypothesis, its least and greatest equilibria can be characterized and, under some assumptions of finiteness, calculated 
as least and greatest fixed points by the well-known lattice-theoretical Knaster-Tarski fixed point theorem, which provides the theoretical basis 
for the Robinson-Topkis algorithm \cite{topkis98}. 

Since the breakthrough on the PPAD-completeness of finding mixed Nash equilibria~\cite{dgp09}, 
the question of approximating Nash equilibria
emerged as a key problem in algorithmic game theory~\cite{dmp07,hk11}. In this context, approximate equilibrium refers to $\epsilon$-approximation, 
with $\epsilon>0$, 
meaning that, for each player, all the strategies have a payoff which is at most $\epsilon$ more (or less) than the precise payoff of the given strategy. 
It is well known that the notion of correct (a.k.a.\ sound) approximation is fundamental in static program analysis, one major research area 
in programming language
theory and design. 
Static program analysis derives some partial but correct information of the run-time program behavior 
without actually executing programs. Prominent examples of static analysis include dataflow analysis used in program compilers,
type systems for inferring program types, model checking for program verification, and abstract interpretation used to design
abstract interpreters of programs. In particular, the abstract interpretation approach to static analysis~\cite{CC77,CC79} 
relies on a lattice-theoretical model of the notion of approximation. Program properties are modelled by 
a domain $C$ endowed with a partial order  $\leq$ which plays the role of approximation relation, where $x\leq y$
intuitively means that the property $y$ is an approximation of the property $x$, or, equivalently, that the property $x$ is 
logically stronger than $y$. 
The key principle in static analysis by abstract interpretation is to provide an approximate interpretation, a.k.a.\ an abstract interpretation,
of a program for a given abstraction of the properties of its concrete semantics. This leads to the idea of abstract domain, 
which is an ordered collection of abstract program properties 
which can be inferred by static analysis,
where approximation is again modeled by the ordering relation.  The classical introductory example of
program abstract interpretation is sign analysis. 
Given an arithmetic integer expression $e$, one tries to bound its sign---negative, zero or positive---without actually computing $e$.
The idea is that one can prove that $e\equiv 3\times -2$ is negative without actually computing that $e$ evaluates to $-6$. 
If $\bS=\{-,0,+\}$ then abstract integers in $A$ are defined as subsets of these signs in $\bS$, i.e., $A\ud \wp(\bS)$. 
Here, $A$ is ordered by inclusion which encodes the approximation relation: for example, $\{+\} \subseteq \{0,+\}$ encodes 
that being positive is a stronger property than being nonnegative, so that nonnegative is an approximation of positive. 
Then, any set of integer numbers $S\in \wp(\bZ)$ can be abstractly represented by
its most precise abstraction in $A$ through an abstraction function $\alpha:\wp(\bZ) \ra A$. Hence, a set of integers $S$ is correctly
approximated by an abstract integer $a\in A$ precisely when $\alpha(S) \subseteq a$ holds. In turn,
one can define abstract addition $\oplus$ and multiplication $\otimes$ on abstract integers in $A$: for example, $\{-,0\} \oplus \{-\}=
\{-\}$ and $\{-\} \oplus \{+\} = \{-,0,+\}$, while $\{-\}\otimes\{+,0\} = \{-,0\}$ and $\{-,+\}\otimes \{0\} = \{0\}$. 
Hence, in order to analyze the expression $3\times -2$ we convert it to  $\alpha(\{3\}) \otimes \alpha(\{-2\})$ to infer $\{-\}$. Of course, 
it may well happen that the abstract domain does not carry enough precision to compute the most precise information theoretically
available in $A$: for the expression $-2 + 2$, we have that $\alpha(\{-2\}) \oplus \alpha(\{2\}) = \{-,0,+\}$ although
$\alpha(\{-2+2\})=\{0\} \subsetneq \alpha(\{-2\}) \oplus \alpha(\{2\})$. In such cases, the output of the static analysis is ``I don't know''. In the terminology
of abstract interpretation, $\oplus$ and $\otimes$ are correct approximations of concrete integer addition and multiplication. 
Program semantics are typically formalized using fixed points of functions for modelling loops and recursive procedures. A
basic result of abstract interpretation tells us that correctness is preserved for least and greatest fixed points: 
if a concrete monotone function $f:C\ra C$ is correctly approximated by an abstract monotone function $f^\sharp:A\ra A$ on an
abstraction $A$ of $C$ then the least (or greatest) fixed point $\lfp(f)\in C$ of $f$ is correctly approximated by the least (or greatest) fixed point 
$\lfp(f^\sharp)\in A$  of $f^\sharp$,
i.e., $\alpha(\lfp(f)) \leq_A \lfp(f^\sharp)$. For example, the concrete output of 
the
program $P\equiv x:=3; \textbf{while}~(x < 13)~\textbf{do}~x:=2*x$ is $\{24\}$, while
its abstract interpretation is derived as the least fixed point which is 
greater than or equal to the initial abstract value $\alpha(\{3\})=\{+\}$
for the function
$f^\sharp:A\ra A$ defined by $f^\sharp (a) = \alpha(\{2\}) \otimes a$, so that this least fixed point is $\lfp_{\geq \{+\}}(f^\sharp)= \{+\}$, 
and in this case we have that $\alpha(\{24\}) =  \lfp_{\geq \{+\}}(f^\sharp)$. 

\paragraph{Goal.}
The similarities between supermodular games and formal program semantics should be therefore clear, since they 
both rely on order-theoretical models and on computing extremal fixed points of suitable functions on lattices. 
However, while the order theory-based 
approximation of program semantics by static analysis
is a traditional and well-established area in computer science since forty years, to the best of our knowledge, no attempt has been made
to apply some techniques used in static program analysis for defining a corresponding notion of approximation in supermodular games. 
The overall goal of this paper is to investigate whether and how abstract interpretation can be used to define and calculate approximate Nash equilibria
of supermodular games, where the key notion of approximation will be modeled by a partial ordering relation similarly to what happens in static program 
analysis.  This appears to be the first contribution to make use of an order-theoretical notion
of approximation for equilibria of supermodular games, in particular by resorting to the abstract interpretation technique
ordinarily used in static program analysis.  

\paragraph{Contributions.}
As sketched above, abstract interpretation essentially relies on: (1)~abstract domains $A$ which encode approximate properties; 
(2)~abstract functions $f^\sharp$ which must correctly approximate on $A$ the behavior of some
concrete operations $f$; (3)~results of correctness for the abstract interpreter using $A$ and $f^\sharp$, for example 
the correctness of extremal fixed points of abstract functions, e.g.\ $\lfp(f^\sharp)$ correctly approximates $\lfp(f)$; 
(4)~so-called widening/narrowing operators tailored for the abstract domains $A$ to ensure and/or accelerate the convergence in iterative 
fixed point computations of abstract functions $f^\sharp$.  We contribute to set up a general framework for designing abstract interpretations
of supermodular games which basically encompasses the above points~(1)-(3), while widening/narrowing operators are not taken into account
since their definition is closely related to some individual abstract domain. Our main contributions can be summarized as follows.

\begin{itemize}
\item In supermodular games, a strategy space $S_i$ for the player $i$ 
is assumed to be a complete lattice and best response correspondences are (multivalued) functions
defined over a product $S_1 \times \cdots \times S_N$ of complete lattices which plays the role
of concrete domain. Thus, as a preliminary step, we show how abstractions of strategy spaces  
can be composed in order to define an abstract domain
of the product $S_1 \times \cdots \times S_N$, and, on the other hand, an abstraction of the product $S_1 \times \cdots \times S_N$ 
can be decomposed into abstract domains of the individual $S_i$'s. 

\item Abstract interpretation is commonly used for approximating single-valued functions on complete lattices. For supermodular games, 
best responses are indeed multivalued functions $B:S_1 \times \cdots \times S_N \ra \wp(S_1 \times \cdots \times S_N)$ that we expect to
approximate. Thus, we first provide short and direct constructive proofs ensuring the existence of fixed points for multivalued
functions. Then, we show how abstract interpretation can be generalized to handle multivalued functions, first by defining a parametric
notion of correct approximation for multivalued functions, and then by proving that these correct abstract multivalued functions preserve
their correctness for their fixed points. 

\item We investigate how to define an ``abstract interpreter'' of a supermodular game. The first approach consists in defining 
a supermodular game on an abstract strategy space. Given a game $\Gamma$ with strategy spaces $S_i$ and
utility functions $u_i: S_1 \times \cdots \times S_N \ra \bR$, this means that we assume a family of abstractions $A_i$, one for each $S_i$, 
that gives rise to an abstract strategy space $A=A_1 \times \cdots \times A_N$, 
and a suitable abstract restriction of the utility functions $u_i^A:A_1 \times \cdots \times A_N \ra \bR$. This defines 
what we call an abstract game $\Gamma^A$, which, under some conditions, has abstract equilibria which correctly approximate the equilibria of
$\Gamma$. Obviously, the fixed point computations over $A$ for the abstract game $\Gamma^A$ should be more efficient than in $\Gamma$. 
This abstraction technique provides a generalization of the efficient algorithm by Echenique~\cite{ech2007} for finding all equilibria in
a finite game with strategic complementarities. 

\item On the other hand, we put forward a second notion of abstract game where the strategy spaces are subject to a kind of
partial approximation, meaning that, for any utility function, we consider approximations of the strategy spaces
of the ``other players'', i.e., correct approximations over abstract domains $A_i$ of the 
functions $u_i(s_i, \cdot): S_1\times \cdots S_{i-1}\times S_{i+1}\times \cdots \times S_N\ra \bR$, for any given strategy $s_i\in S_i$.
This abstraction technique gives rise to games having an abstract best response correspondence. 
This approach is inspired and somehow generalizes the implicit methodology of approximate computation of 
equilibria 
considered by Carl and Heikkil{\"{a}} \cite[Chapter~8]{ch11}.
\end{itemize}

\noindent
Our results are illustrated on some examples of supermodular games, in particular a couple of examples of 
Bertrand oligopoly models are taken from Carl and Heikkil{\"{a}}'s book \cite{ch11}.

\section{Background}

\subsection{Order-Theoretical Notions}
Given a function $f : X \ra Y$ and a subset $S\subseteq X$ then 
$\ok f(S) \ud \{f(s)\in Y~|~s\in S\}$ denotes the image of $f$ on $S$ and
$f^s:\wp(X)\ra \wp(Y)$ denotes the corresponding standard powerset lifting of $f$,
that is,  $f^s (S) \ud f(S)$. 
Given a family of $N>0$ sets $(S_i)_{i=1}^N$, $\times_{i=1}^N S_i$ denotes their Cartesian product. If $i\in [1,N]$ and 
$s\in \times_{i=1}^N S_i$ then $S_{-i} \ud S_1 \times \cdots \times S_{i-1}\times S_{i+1}\times \cdots S_{N}$, while
$s_{-i}\ud (s_1,\ldots, s_{i-1},s_{i+1},\ldots,s_{N})\in S_{-i}$. 
Also, $\tuple{\bR^N,\leq}$ denotes the standard product poset of real numbers, where for $s,t\in \bR^N$, $s\leq t$ iff 
for any $i\in [1,N]$, $s_i\leq t_i$, while $s+t=(s_i+t_i)_{i=1}^N$. 
A multivalued function, also called correspondence, is a mapping $f:X\ra \wp(X)$. An element $x\in X$ is a fixed point of $f$ when $x\in f(x)$,
where $\Fix(f)\ud \{x\in X~|~ x\in f(x)\}$.

Let $\tuple{C,\leq,\wedge,\vee,\bot,\top}$ be a complete lattice, compactly denoted  by $\tuple{C,\leq}$.  
A nonempty subset $S\subseteq C$ is a subcomplete sublattice of $C$ if for all its nonempty subsets $X\subseteq S$, $\wedge X\in S$ 
and $\vee X \in S$. 
Let us recall the following relations on the powerset $\wp(C)$:
for any $X,Y\in \wp(C)$,
\begin{align*}
\text{(Smyth preorder)} \qquad & X\preceq_S Y \;\triangleRel\; \forall y\in Y.\exists x\in X.\: x\leq y \\
\text{(Hoare preorder)}\qquad & X\preceq_H Y \;\triangleRel\; \forall x\in X.\exists y\in Y.\: x\leq y \\
\text{(Egli-Milner preorder)} \qquad& X\pem Y \;\triangleRel\; X\preceq_S Y \;\:\&\;\: X\preceq_H Y \\
\text{(Veinott relation)}\qquad & X\preceq_V Y \;\triangleRel\; \forall x\in X.\forall y\in Y.\: x\wedge y\in X \;\:\&\;\: x\vee y \in Y
\end{align*}
Smyth, Hoare and Egli-Milner relations are preorders (i.e., reflexive and transitive), 
while Veinott relation (also called strong set relation) is transitive and antisymmetric. 
A multivalued function $f:C\ra \wp(C)$ is $S$-monotone if for any $x,y\in C$, $x\leq y$ implies $f(x) \preceq_S f(y)$. 
$H$-, $\EM$- and $V$-monotonicity are defined analogously. 
We also use the following notations:
\[
\begin{array}{c}
\wp^{\scriptscriptstyle \wedge}(C) \ud \{ X\in \wp(C)~|~ \wedge\! X \in X\}
\qquad \wpvee(C) \ud \{ X\in \wp(C)~|~ \vee\! X \in X\}\\[5pt]
\wp^\diamond(C)  \ud \wpwedge(C) \cap \wpvee(C) 
\qquad \SL(C)  \ud \{X\in\wp(C) ~|~ X \neq \varnothing,\, X \text{~subcomplete sublattice of~} C\}
\end{array}
\]
Observe that if $X,Y\in \wpwedge(C)$ then $X\preceq_S Y\:\Lra\: \wedge X \leq \wedge Y$. Similarly, 
if $X,Y\in \wpvee(C)$ then $X\preceq_H Y\:\Lra\: \vee X \leq \vee Y$ and 
if $X,Y\in \wp^\diamond(C)$ then $X\preceq_{\EM} Y\:\Lra\: \wedge X \leq \wedge Y \;\&\: \vee\!\! X \leq \vee Y$.

The pointwise ordering relation $\sqsubseteq$ between two functions $f,g: X \ra C$ whose range $C$ is a complete lattice,  
is defined by $f\sqsubseteq g$ if for any $x\in X$, $f(x)\leq_C g(x)$.  
A function $f:C\ra D$ between complete lattices is additive (co-additive) when $f$ preserves arbitrary lub's (glb's).
Given a function $f:C\ra C$ on a complete lattice $C$, $\Fix(f)\ud \{x\in C~|~x=f(x)\}$ denotes the set of fixed points of $f$, while
$\lfp(f)$ and $\gfp(f)$ denote, respectively, the least and greatest fixed points of $f$, when they exist. Let us recall that
least and greatest fixed points always exist for monotone functions.  
If $f:C\ra C$ then for any ordinal $\alpha\in \bO$, the $\alpha$-power $f^\alpha: C\ra C$ is defined by transfinite induction
as follows:
for any $x\in C$,  (1)~if $\alpha=0$ then $f^0 (x) \ud x$; (2)~if $\alpha=\beta+1$ then $f^{\beta+1} (x) \ud f(f^\beta (x))$; 
(3)~if $\alpha = \vee \{\beta\in \bO~|~\beta<\alpha\}$ then $f^\alpha(x) \ud \bigvee_{\beta <\alpha} f^\beta(x)$.

A map $\rho:C\ra C$, with $C$ complete lattice, 
is a (topological) closure operator when: (i)~$x\leq y\:\Ra\: \rho(x)\leq\rho(y)$; (ii)~$x\leq\rho(x)$;
(iii)~$\rho(\rho(x))=\rho(x)$. We denote by $\uco(\tuple{C,\leq})$ the set of all closure operators on the complete lattice $C$.   
A closure operator $\rho\in \uco(C)$ is uniquely determined by its image $\rho(C)$, which coincides with its set of fixed points $\Fix(\rho)$, 
as follows: for any $c\in C$, $\rho(c) =\wedge_C \{x\in \rho(C)~|~ c\leq x\}$. Also, a subset $S\subseteq C$ is the image of a closure
operator $\rho_S\in \uco(C)$ iff $S$ is meet-closed, i.e., $S=\{\wedge_C X \in C~|~ X\subseteq S\}$; in this case, $\rho_S(c)= \wedge_C \{x\in S~|~ c\leq x\}$.

\paragraph*{Supermodularity.} 
Given a complete lattice $C$,
a function $u:C\ra \bR^N$ is a supermodular if for any $c_1,c_2\in C$, $u(c_1 \vee c_2) + u(c_1 \wedge c_2) \geq u(c_1) + u(c_2)$, while
$u$ is  quasisupermodular if for any $c_1,c_2\in C$, 
$u(c_1 \wedge c_2)\leq u(c_1) \:\Ra\: u(c_2) \leq  u(c_1 \vee c_2)$ and 
$u(c_1 \wedge c_2)< u(c_1) \:\Ra\: u(c_2) <  u(c_1 \vee c_2)$.
Clearly, supermodularity implies quasisupermodularity (while the converse is not true).  
Recall that if $u:C\ra \bR^N$ is quasisupermodular then $\argmax(f) \ud \{x\in C~|~ \forall y\in C.\: f(y)\leq f(x)\}$ is a sublattice
of $C$. 

A function $u:C_1\times C_2 \ra \bR^N$ has increasing differences when for any $(x,y)\leq (x',y')$, 
$u(x',y) - u(x,y) \leq u(x',y') - u(x,y')$, or, equivalently, the functions $u(x',\cdot)-u(x,\cdot)$ and $u(\cdot,y')-u(\cdot,y)$ are monotone. 
A function $u:C_1\times C_2 \ra \bR^N$
has the single crossing property when for any $(x,y)\leq (x',y')$, 
$u(x,y)\leq u(x',y) \:\Ra\: u(x,y') \leq u(x',y')$ and $u(x,y)< u(x',y) \:\Ra\: u(x,y') < u(x',y')$.
Clearly, if $u$ has increasing differences then $u$ has the single crossing property, while the converse does not hold. 

Supermodularity on product complete lattices and increasing differences are related as follows:
a function  $u:C_1\times C_2\ra \bR^N$ is supermodular if and only if 
$u$ 
has increasing differences and, for any $c_i\in C_i$, $u(c_1,\cdot):C_2 \ra \bR^N$
and $u(\cdot,c_2):C_1 \ra \bR^N$ are supermodular.

\subsection{Noncooperative Games}
In our model, a noncooperative game 
$\Gamma = \langle S_i,u_i \rangle_{i=1}^{n}$ for players $i=1,...,n$ consists of a family of feasible strategy spaces $(S_i,\leq_i)_{i=1}^n$ 
which are assumed to be complete lattices, so that 
the strategy space $S\ud \times_{i=1}^{n} S_i$ is a complete lattice for the componentwise order $\leq$,  and of 
a family of utility (or payoff) functions $u_i:S\ra \bR^{N_i}$, with $N_i\geq 1$.
The $i$-th best response correspondence $B_i:S_{-i} \ra \wp(S_i)$
is defined as $B_{i}(s_{-i}) \ud \{x_i\in S_i~|~ \forall s_i\in S_i.\: u_i(s_i,s_{-i}) \leq u_i(x_i,s_{-i})\}$,
while the best response correspondence $B: S\ra \wp(S)$ is defined by
$B(s_1,...,s_n) \ud \times_{i=1}^{n} B_i(s_{-i})$. 
A strategy $s\in S$ is a pure Nash equilibrium for $\Gamma$ when $s$ is a fixed point of $B$, i.e., 
$s\in B(s)$, meaning that in  $s$ 
there is no feasible way for any player to strictly improve its utility if the strategies of all the other players 
remain unchanged.
 We denote by $\Eq(\Gamma)\in \wp(S)$ the
set of Nash equilibria for $\Gamma$, so that $\Eq(\Gamma)=\Fix(B)$. 

\subsubsection{(Quasi)Supermodular Games}\label{qsg}
A noncooperative $\Gamma = \tuple{S_i,u_i}_{i=1}^{n}$ is \emph{supermodular} when: 

\medskip
(1)~for any $i$, for any $s_{-i}\in S_{-i}$, $u_i(\cdot,s_{-i}):S_i\ra \bR^{N_i}$ is supermodular; 

(2)~for any $i$, $u_i(\cdot,\cdot):S_i\times S_{-i} \ra \bR^{N_i}$ has increasing differences. 

\medskip
\noindent
On the other hand, $\Gamma$ is \emph{quasisupermodular} (or, with strategic complementarities) when: 

\medskip
(1)~for any $i$, for any $s_{-i}\in S_{-i}$, $u_i(\cdot,s_{-i}):S_i\ra \bR^{N_i}$ is quasisupermodular;

(2)~for any $i$, $u_i(\cdot,\cdot):S_i\times S_{-i} \ra \bR^{N_i}$ has the single crossing property. 

\medskip
\noindent
In these cases, it turns out (cf.\
\cite[Theorems~2.8.1 and~2.8.6]{topkis98}) that 
the $i$-th best response correspondence $B_i:S_{-i} \ra \wp(S_i)$ is $\EM$-monotone, as well as 
the best response correspondence $B: S\ra \wp(S)$. 

Let us recall that, given  a complete lattice $C$, a 
function $f:C\ra \bR^N$ is order upper semicontinuous if for any chain $Y\subseteq C$, 
\[
\textstyle
\limsup\limits_{x\in Y, x\ra \vee Y} f(x) \leq f(\vee C) \quad\text{and}\quad \limsup\limits_{x\in Y, x\ra \wedge Y} f(x) \leq f(\wedge C).
\]
It turns out (cf.\ \cite[Lemma~4.2.2]{topkis98}) that if 
each $u_i(\cdot,s_{-i}):S_i\ra \bR^{N_i}$ is 
order upper semicontinuous then, for each $s\in S$,  
$B_{i}(s_{-i})\in \SL(S_i)$, i.e., $B_{i}(s_{-i})$ is a nonempty subcomplete sublattice of $S_i$, so that
$B(s) \in \SL(S)$ also holds. In particular, we have that $\wedge_i B_{i}(s_{-i}), \vee_i B_{i}(s_{-i})
\in B_{i}(s_{-i})$ as well as $\wedge B(s), \vee B(s) \in B(s)$, namely, 
$B_i(s_{-i})\in \wp^\diamond(S_i)$ and $B(s) \in \wp^\diamond(S)$.  
It also turns out \cite[Theorem~2]{zhou94} that  $\tuple{\Eq(\Gamma),\leq}$ is a complete lattice---although, in general, 
it is not a subcomplete sublattice of $S$---and therefore $\Gamma$ admits 
the least and greatest Nash equilibria, which are denoted, respectively, 
by $\lne(\Gamma)$ and $\gne(\Gamma)$. 
It should be remarked that the hypothesis of upper semicontinuity for $u_i(\cdot,s_{-i})$ holds 
for any finite-strategy game, namely for those games where each strategy space $S_i$ is finite. In the
following, we will consider (quasi)supermodular games which satisfy this hypothesis of upper semicontinuity. 

If, given any $s_i\in S_i$, the function $u_i(s_i,\cdot):S_{-i}\ra \bR^{N_i}$ is monotone then it turns out 
\cite[Propositions~8.23 and 8.51]{ch11}
that 
$\gne(\Gamma)$ majorizes all equilibria, i.e., for all $i$ and $s\in \Eq(\Gamma)$, $u_i(\gne(\Gamma)) \geq u_i (s)$, while
$\lne(\Gamma)$ minimizes all equilibria.

\subsection{Computing Game Equilibria}\label{comp-sec}
Consider a (quasi)supermodular game $\Gamma = \tuple{S_i,u_i}_{i=1}^{n}$ and define the
functions $B_{\wedge}, B_{\vee}:S\ra S$ as follows: $B_{\wedge}(s) \ud \wedge B(s)$
and $B_{\vee}(s) \ud \vee B(s)$. As recalled in Section~\ref{qsg}, we have that $B_{\wedge}(s), B_{\vee}(s)
\in B(s)$. When the image of the strategy space $S$ for $B_{\wedge}$ turns out to be finite, 
the standard algorithm \cite[Algorithm~4.3.2]{topkis98} for computing $\lne(\Gamma)$ consists in applying the constructive 
Knaster-Tarski fixed point theorem 
to the function $B_{\wedge}$ so that
$\lne(\Gamma) = \bigvee_{k\geq 0} B_{\wedge}^k (\bot_S)$. Dually, 
we have that $\gne(\Gamma) = \bigwedge_{k\geq 0} B_{\vee}^k (\top_S)$. In particular, this procedure can be always used for finite
games. The application of the so-called chaotic iteration in this fixed point computation 
provides the 
Robinson-Topkis (RT) algorithm \cite[Algorithm~4.3.1]{topkis98} in Figure~\ref{fig-algo}, also called round-robin optimization, which is
presented in its version for least fixed points, while the statements in comments provide the version for calculating greatest fixed points. 

\begin{figure}\label{fig-algo}
\begin{flalign*}
&\tuple{s_1,...,s_n} := \tuple{\bot_1,...,\bot_n}; \quad \mathtt{//}\tuple{s_1,...,s_n} := \tuple{\top_1,...,\top_n}; \\
&\textbf{do}~ \big\{\tuple{t_1,...,t_n} := \tuple{s_1,...,s_n};\\
&\qquad s_1 := \wedge_1 B_1(s_{-1}); \;\quad\quad\;\, \mathtt{//}s_1 := \vee_1 B_1(s_{-1}); \\
&\qquad \ldots\\
&\qquad s_n := \wedge_n B_n(s_{-n}); \;\quad\quad \mathtt{//}s_1 := \vee_n B_n(s_{-n});\\ 
&\big\}\\
&\textbf{while}~ \neg (\tuple{s_1,...,s_n} = \tuple{t_1,...,t_n})
\end{flalign*}
\caption{Robinson-Topkis (RT) algorithm.}
\end{figure}

Let us provide a running example of supermodular finite game.  
\begin{example}\label{ex1}\rm
Consider a two players finite game $\Gamma$ represented in normal form by the following double-entry payoff matrix:

\begin{center}
\renewcommand{\arraystretch}{1.2}
\begin{tabular}{c|C{25pt}|C{25pt}|C{25pt}|C{25pt}|C{25pt}|C{25pt}|}
\multicolumn{1}{c}{}
& \multicolumn{1}{c}{\textbf{1}} &  \multicolumn{1}{c}{\textbf{2}} &  \multicolumn{1}{c}{\textbf{3}} &  \multicolumn{1}{c}{\textbf{4}} &  \multicolumn{1}{c}{\textbf{5}} &  \multicolumn{1}{c}{\textbf{6}} \\
\cline{2-7}
\textbf{6} &-1,\,-3 &-1,\,-1 &2,\,4 &5,\,6 &6,\,5 &6,\,5 \\
\cline{2-7}
\textbf{5} & 0,\,0 & 0,\,2 & 3,\,4 & 6,\,6 &7,\,5 & 6,\,5\\
\cline{2-7}
\textbf{4} &3,\,1  & 3,\,3 & 3,\,5 & 5,\,6 & 5,\,5 & 4,\,4 \\
\cline{2-7}
\textbf{3} & 2,\,2 & 2,\,4 & 2,\,6 & 4,\,5 & 4,\,4 &3,\,2 \\
\cline{2-7}
\textbf{2} & 6,\,4 & 6,\,6 & 6,\,7 & 6,\,4 & 5,\,2 & 4,\,-1 \\
\cline{2-7}
\textbf{1} & 6,\,4 & 5,\,6 & 5,\,6 & 4,\,2 & 3,\,0 & 2,\,-3 \\
\cline{2-7}
\end{tabular}
\end{center}

\noindent
Here, $S_1$ and $S_2$ are both the finite chain of integers $C=\tuple{\{1,2,3,4,5,6\},\leq}$ and $u_1(x,y),u_2(x,y):S_1\times S_2 \ra \bR$ are, respectively, 
the first and second entry in the matrix element determined by row $x$ and column $y$. It turns out that both $u_1$ and $u_2$ 
have increasing differences, so that, since $S_1$ and $S_2$ are finite chains, $\Gamma$ is a finite supermodular game. 
The two best response correspondences $B_1,B_2:C \ra \SL(C)$ are as follows: 
\[
\begin{array}{llllll}
B_1(1)=\{1,2\},& B_1(2)=\{2\},& B_1(3)=\{2\},& B_1(4)=\{2,5\},& B_1(5)=\{5\},& B_1(6)=\{5,6\};\\
B_2(1)=\{2,3\},& B_2(2)=\{3\},& B_2(3)=\{3\},& B_2(4)=\{4\},& B_2(5)=\{4\},& B_2(6)=\{4\}.
\end{array}
\]
Thus, $\Eq(\Gamma) = \{(2,3),(5,4)\}$, since this  is  the set $\Fix(B)$ of fixed points of the best response correspondence
$B=B_1 \times B_2$. We also notice that $u_1(\cdot,s_2),u_2(s_1,\cdot): C \ra \bR$ are neither monotone nor antimonotone. 
The fixed point computations of the least and greatest equilibria through the above RT algorithm proceed as follows: 
\[
\begin{array}{l}
(1,1) \mapsto \big(\!\wedge \!B_1(1,1),1\big) = (1,1) \mapsto \big(1,\wedge B_2(1,1)\big)=(1,2) \mapsto (2,2) \mapsto (2,3) \mapsto (2,3) \mapsto (2,3)\quad \text{(lfp)}\\[5pt]
(6,6) \mapsto (\vee B_1(6,6),6) = (6,6) \mapsto (6,\vee B_2(6,6))=(6,4) \mapsto (5,4)\mapsto (5,4)\mapsto (5,4) \quad\text{(gfp)}\qquad\qed
\end{array}
\]
\end{example}

\subsection{Abstract Interpretation}\label{ai-sec}
Static program analysis relies on correct (a.k.a.\ sound) and computable semantic 
approximations. A program $P$  
is modeled
by some semantics $\Sem\!\grasse{P}$ and a
static analysis of $P$ is designed as an approximate 
semantics 
$\Sem^\sharp \!\grasse{P}$  
which must be correct w.r.t.\ $\Sem\!\grasse{P}$. 
This may be called global correctness of static analysis. Any (finite) program $P$ 
is a suitable composition of a number of
constituents subprograms $c_i$ and this is reflected on its global 
semantics $\Sem\!\grasse{P}$ which is commonly defined 
by some combinations of the semantics $\Sem\!\grasse{c_i}$ of its
components.  
Thus, global correctness of a static 
analysis of $P$ 
is typically derived from local correctness of static analyses for
its components $c_i$. This global vs.\ local picture of  
static analysis correctness is very common, independently of the kind of programs (imperative, functional, reactive, etc.),
of static analysis techniques (model checking, abstract interpretation, logical deductive systems, 
type systems, etc.), of program properties under analysis 
(safety, liveness, numerical properties, pointer aliasing, type safety, etc.).
A basic and rough proof principle in static analysis is that
global correctness is derived from local correctness. 
In particular this applies 
to static program analyses that are designed using some form of abstract 
interpretation. 
Let us consider a simplified but recurrent scenario, where 
$\Sem\!\grasse{P}$
is defined as least (or greatest) fixed point $\lfp(f)$ of a monotone function $f$ on
some domain $C$ of program properties, 
 which is endowed with a partial order that encodes
the relative precision of properties. In abstract interpretation, 
a static analysis is then
specified as an abstract fixed point computation which must be correct for $\lfp(f)$. This is routinely defined 
through an ordered abstract domain $A$ of properties and an abstract 
semantic function $f^\sharp : A\ra A$ that give rise to a fixed point-based
static analysis $\lfp(f^\sharp)$ (whose decidability and/or practical 
scalability is usually ensured by chain conditions on $A$, widenings/narrowings operators, 
interpolations, etc.). Correctness relies on encoding approximation through
a concretization map $\gamma: A \ra C$ and/or an abstraction map $\alpha:
C\ra A$: the approximation of some value $c$ through an abstract property 
$a$ is encoded as $\alpha(c) \leq_A a$ or~---~equivalently, when
$\alpha$/$\gamma$ form a Galois connection~---~$c \leq_C \gamma(a)$. 
Hence, global correctness translates to $\alpha(\lfp(f)) \leq \lfp(f^\sharp)$, 
local correctness means $\alpha \circ f \sqsubseteq f^\sharp \circ \alpha$, and 
the well-known ``fixed point approximation lemma''~\cite{CC77,CC79}
tells us that local implies global correctness. 

In standard abstract interpretation~\cite{CC77,CC79}, 
abstract domains, also called abstractions, 
are specified by Galois connections/insertions
(GCs/GIs for short). 
Concrete and abstract domains, $\tuple{C,\leq_C}$ and
$\tuple{A,\leq_A}$, are assumed to be complete lattices 
which are related by abstraction and concretization maps
$\alpha:C\ra A$ and $\gamma:A \ra C$ that give
rise to a GC $(\alpha,C,A,\gamma)$, that is,
for all $a\in A$ and $c\in C$,
$\alpha(c) \leq_A a \Lra c \leq_C \gamma(a)$. 
A GC is a GI when $\alpha\circ\gamma=\id$. 
A GC is (finitely) disjunctive when $\gamma$ preserves all (finite) lubs. 
We use $\Abs(C)$ to denote all the possible abstractions of $C$, where
$A\in \Abs(C)$ means that $A$ is an abstract domain of $C$ specified by some GC/GI.
%
Let us recall some well known properties of a GC
$(\alpha,C,A,\gamma)$: (1)~$\alpha$ is additive; 
(2)~$\gamma$ is co-additive; 
(3)~$\gamma\circ \alpha: C\ra C$ is a closure operator;
(4)~if $\rho:C\ra C$ is a closure operator
then $(\rho,C,\rho(C),\id)$ is a GI; (5)~$(\alpha,C,A,\gamma)$ is a GC iff $\gamma(A)$ is the image of 
a closure operator on $C$; (6)~a GC $(\alpha,C,A,\gamma)$ is (finitely) disjunctive iff $\gamma(A)$ is (finitely)
meet- and join-closed.

\begin{example}\label{ex-chains}\rm 
Let us consider a concrete domain $\tuple{C,\leq}$ which is a finite chain. Then, it turns out that 
$(\alpha,C,A,\gamma)$ is a GC iff $\gamma(A)$ is the image of a closure operator on $C$ iff $\gamma(A)$ 
is a any subset of $C$ which contains $\top_C$. As an example, for 
the game $\Gamma$ in Example~\ref{ex1}, where $S_i$ is the chain of integers $[1,6]$, we have that 
$A_1 =\{3,5,6\}$ and 
$A_2=\{2,6\}$ are two abstractions of $C$. 
\qed
\end{example}

\begin{example}\label{ex-ceil}\rm 
Let us consider the ceil function on real numbers $\ceil{\cdot}:\bR \ra \bR$, that is, $\ceil{x}$ is 
the smallest integer not less than $x$. Let us observe that $\ceil{\cdot}$ is a closure operator on $\tuple{\bR,\leq}$ because:
(1)~$x\leq y\:\Ra\: \ceil{x} \leq \ceil{y}$; (2)~$x\leq \ceil{x}$; (3)~$\ceil{\ceil{x}}=\ceil{x}$. Therefore, the ceil
function allows us to view integer numbers $\bZ=\ceil{\bR}$ as an abstraction of real numbers. The ceil function
can be generalized to any finite fractional part of real numbers:  given any integer number $N\geq 0$, $\cl_N:\bR\ra \bR$ is
defined as follows: $\cl_N(x) = \frac{\ceil{10^N x}}{10^N}$. For $N=0$, $\cl_N(x)=\ceil{x}$, while for $N>0$, $\cl_N(x)$ is 
the smallest rational number with at most $N$ fractional digits not less than $x$.  For example, 
if $x\in \bR$ and $1<x\leq 1.01$ then $\cl_2(x) = 1.01$. 
Clearly, it turns out that $\cl_N$ is
a closure operator which permits to cast rational numbers with at most $N$ fractional digits as an abstraction of real numbers. 
\qed
\end{example}

Let $f:C\ra C$ be some 
concrete monotone function---to keep notation simple, 
we consider 1-ary functions---and let
 $\ok{f^\sharp:A \ra A}$ be a corresponding monotone abstract function
 defined on some abstraction $A$ specified by a GC $(\alpha,C,A,\gamma)$. Then,
$\ok{f^\sharp}$ is a correct (or sound) approximation of $f$ on $A$
when $\ok{f \circ \gamma \sqsubseteq \gamma \circ f^\sharp}$ holds. 
If $\ok{f^\sharp}$ is a correct approximation of $f$ then we also have fixed point correctness, that is, 
$\lfp(f) \leq_C  \gamma(\lfp(\ok{f^\sharp}))$ and  $\gfp(f) \leq_C  \gamma(\gfp(\ok{f^\sharp}))$.
The abstract function
$\ok{f^A \ud \alpha \circ f \circ \gamma: A\rightarrow A}$ 
is called the best
correct approximation of $f$ on $A$, because any abstract
function
$\ok{f^\sharp}$ is correct iff $\ok{f^A \sqsubseteq f^\sharp}$.
Hence,
$\ok{f^A}$ plays the role of the 
best possible approximation of $f$ on the abstraction~$A$.

\section{Abstractions on Product Domains}
Let us show how abstractions of different concrete domains $C_i$ can be composed in order to define an abstract domain
of the product domain $\times_i C_i$, and, on the other hand, an abstraction of a product $\times_i C_i$ can be decomposed into
abstract domains of the component domains $C_i$. In the following, we consider a  finite family of complete
lattices $\tuple{C_i,\leq_i}_{i=1}^n$, while product domains are considered with the componentwise ordering relation. 

\paragraph{Product Composition of Abstractions.}
This method has been introduced by Cousot and Cousot in \cite[Section~4.4]{cc94}.
Given a family of GCs $(\alpha_i,C_i,A_i,\gamma_i)_{i=1}^n$, one can easily define a componentwise abstraction 
$(\alpha,\times_{i=1}^n C_i,\times_{i=1}^n A_i,\gamma)$ 
of the product complete lattice
$\times_{i=1}^n C_i$, where 
$\times_{i=1}^n C_i$ and $\times_{i=1}^n A_i$  are both complete lattices w.r.t.\ the componentwise partial order and
for any $c \in \times_{i=1}^n C_i$ and $a\in \times_{i=1}^n A_i$,
$$\alpha(c) \ud (\alpha_i(c_i))_{i=1}^n, \qquad\qquad\gamma (a) \ud (\gamma_i(a_i))_{i=1}^n.$$
For any $i$, we also use the function $\gamma_{-i}:A_{-i} \ra C_{-i}$ to denote
$\gamma_{-i}(a_{-i}) = \gamma(a)_{-i}=(\gamma_j(a_j))_{j\neq i}$.

\begin{lemma}\label{prod2}
$(\alpha,\times_{i=1}^n C_i,\times_{i=1}^n A_i,\gamma)$ is a GC. Moreover, if each $(\alpha_i,C_i,A_i,\gamma_i)$ is a (finitely) disjunctive GC
then $(\alpha,\times_{i=1}^n C_i,\times_{i=1}^n A_i,\gamma)$ is a (finitely) disjunctive GC.
\end{lemma}
\noindent
In static program analysis, 
$(\alpha,\times_{i=1}^n C_i,\times_{i=1}^n A_i,\gamma)$ is called a nonrelational abstraction since, intuitively, 
the product abstraction $\times_{i=1}^n A_i$ 
 does not take into account any relationship between the different concrete domains $C_i$.

\paragraph{Decomposition of Product Abstractions.}
Let us show that any GC $(\alpha,\times_{i=1}^n C_i,A,\gamma)$ for the concrete product domain $\times_{i=1}^n C_i$ induces 
corresponding abstractions $(\alpha_i,C_i,A_i,\gamma_i)$ of $C_i$ as follows:
\begin{itemize}
\item[--] $A_i \ud \{c_i\in C_i~|~\exists a\in A. \gamma(a)_i =c_i\}\subseteq C_i$,
endowed with the partial order $\leq_i$ of $C_i$;
\item[--] for any $c_i\in C_i$, $\alpha_i(c_i) \ud \gamma(\alpha(c_i,\bot_{-i}))_i$;
\item[--] for any $x_i\in A_i$, $\gamma_i (x_i) \ud x_i$.
\end{itemize}

\begin{lemma}\label{prod1}
$(\alpha_i,C_i,A_i,\gamma_i)$  is a GC. Moreover, this GC is (finitely) disjunctive when $(\alpha,\times_{i=1}^n C_i,A,\gamma)$ is (finitely) disjunctive.  
\end{lemma}
\begin{proof}
Let us show that $A_i\subseteq C_i$ is meet-closed. If $X\subseteq A_i$ then for any $x\in X$ there exists some $a_x\in A$ 
such that $\gamma(a_x)_i = x$. Then, let $a\ud \wedge_A \{a_x \in A~|~ x\in X\}\in A$. Since $\gamma$ preserves arbitrary meets,
we have that $\gamma(a) = \wedge_C \{\gamma(a_x)\in C~|~x\in X\}$, so that $\gamma(a)_i = \wedge_{C_i} X$, that is, $\wedge_{C_i} X \in A_i$. 
Hence, since $A_i$ is a Moore-family of $C_i$, we have that $\gamma_i = \mathrm{id}:A_i \ra C_i$ preserves arbitrary meets and therefore is a concretization
function. Let us  check that $\alpha_i$ is the left adjoint of $\gamma_i$, i.e., for any $c_i\in C_i$, 
$\alpha_i(c_i) = \gamma(\alpha(c_i,\bot_{-i}))_i= \wedge_{C_i} \{x_i\in A_i~|~ c_i \leq_i x_i\}$. On the one hand, 
since $(c_i,\bot_{-i}) \leq \gamma(\alpha(c_i,\bot_{-i}))$, we have that
$c_i \leq \gamma(\alpha(c_i,\bot_{-i}))_i$, so that since $\gamma(\alpha(c_i,\bot_{-i}))_i \in A_i$, we conclude that
$\wedge_{C_i} \{x_i\in A_i~|~ c_i \leq_i x_i\} \leq_i \gamma(\alpha(c_i,\bot_{-i}))_i$. On the other hand, 
if $x_i \in A_i$ and $c_i \leq_i x_i$ then $x_i = \gamma(a)_i$ for some $a\in A$, so that we have that 
$(c_i,\bot_{-i}) \leq \gamma(a)$, therefore $\gamma(\alpha(c_i,\bot_{-i})) \leq \gamma(\alpha(\gamma(a)))=\gamma(a)$, and, in turn, 
$\gamma(\alpha(c_i,\bot_{-i}))_i \leq_i \gamma(a)_i = x_i$, which implies that $\gamma(\alpha(c_i,\bot_{-i}))_i \leq_i 
\wedge_{C_i} \{x_i\in A_i~|~ c_i \leq_i x_i\}$. Finally, let us observe that if $\gamma$ is (finitely) additive and
$X\subseteq A_i$ so that for any $x\in X$ there exists some $a_x\in A$ 
such that $\gamma(a_x)_i = x$ then $\gamma(\vee_A \{a_x\in A~|~ x\in X\}) = \vee \{\gamma(a_x) \in \times_{i=1}^n C_i ~|~ x\in X\}$,
so that $\gamma(\vee_A \{a_x\in A~|~ x\in X\})_i = \vee_i \gamma(a_x)_i =\vee_i X$, namely, $\vee_i X\in A_i$, meaning that 
$\gamma_i = \mathrm{id}$ is (finitely) additive. 
\end{proof}

A GC $(\alpha,\times_{i=1}^n C_i,A,\gamma)$ is called \emph{nonrelational} when it is isomorphic to 
the product composition, according to Lemma~\ref{prod2}, of its components  
obtained by Lemma~\ref{prod1}. Of course, the product composition by Lemma~\ref{prod2} of abstract domains is trivially nonrelational. 
Otherwise, $(\alpha,\times_{i=1}^n C_i,A,\gamma)$ is called \emph{relational}. 
It is worth remarking that if $A$ is relational then 
$A$ cannot be 
obtained as a product of abstractions of $C$. As a consequence, 
the relationality of an abstraction $A$
prevents  the definition of a standard noncooperative game over the strategy
space $A$ since $A$ cannot be obtained as a product domain.

\begin{example}\label{ex-comp}\rm
Let us consider the game $\Gamma$ in Example~\ref{ex1} whose finite 
strategy space is $C\times C$, where $C=\{1,2,3,4,5,6\}$ is a chain.  
Consider the subset $A\subseteq C\times C$ as depicted by the following diagram where the ordering is induced 
from $C\times C$:
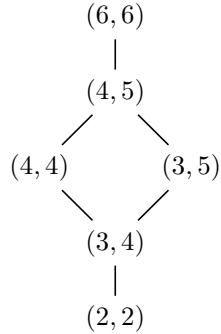
\begin{center}
    \begin{tikzpicture}[scale=0.5]
      \draw (0,0) node[name=22] {{$(2,2)$}};
     \draw (0,2) node[name=34] {{$(3,4)$}};
       \draw (-2,4) node[name=44] {{$(4,4)$}};
       \draw (2,4) node[name=35] {{$(3,5)$}};  
      \draw (0,6) node[name=45] {{$(4,5)$}};
      \draw (0,8) node[name=66] {{$(6,6)$}};

      \draw[semithick] (22) -- (34);
      \draw[semithick] (34) -- (44);
      \draw[semithick] (34) -- (35);
      \draw[semithick] (44) -- (45);
      \draw[semithick] (35) -- (45);
      \draw[semithick] (45) -- (66);

\end{tikzpicture}
\end{center}
Since $A$ is meet- and join-closed and includes the greatest element $(6,6)$ of $C\times C$, we have that $A$ is a disjunctive abstraction
of $C\times C$, where $\alpha:C\times C\ra A$ is the closure operator
induced by $A$ and $\gamma:A\ra C\times C$ is the identity. Observe that $A$ is relational since 
its decomposition by Lemma~\ref{prod1} provides $A_1=\{2,3,4,6\}$ and $A_2=\{2,4,5,6\}$, and the product composition $A_1\times A_2$
by Lemma~\ref{prod2} yields a more expressive abstraction than $A$, for example $(2,4) \in (A_1\times A_2)\smallsetminus A$. 

On the other hand, for the abstractions $A_1 =\{3,5,6\}$ and 
$A_2=\{2,6\}$ of Example~\ref{ex-chains}, by Lemma~\ref{prod2}, the product domain $A_1\times A_2$ is a nonrelational
abstraction of $C\times C$.  
\qed
\end{example}

\section{Approximation of Multivalued Functions}
Let us show how abstract interpretation can be applied to approximate least and greatest fixed points of multivalued functions.  

\subsection{Constructive Results for Fixed Points of Multivalued Functions}
Let $C$ be a complete lattice, $f:C\ra \wp(C)$ be a multivalued function and 
$\fwedge, \fvee: C\ra C$ be the functions defined as: $\fwedge (c)\ud  \wedge f(c)$ and  $\fvee (c)
\ud \vee f(c)$. 
The following constructive result ensuring the existence of least fixed points for a multivalued function 
is given in \cite[Propositions~3.10 and~3.24]{straccia2008}. We provide here a shorter and more direct constructive proof
than in \cite{straccia2008} which is based on the constructive version
of Tarski's fixed point theorem given by Cousot and Cousot~\cite{cc79pjm}.

\begin{lemma}\label{lemmalfp}
If $f:C\ra \wpwedge(C)$ is $S$-monotone then $f$ has the least fixed point $\lfp(f)$. 
Moreover, $\lfp(f) = \bigvee_{\alpha \in \bO} \fwedge^\alpha (\bot)$. 
\end{lemma}
\begin{proof}
By hypothesis, $f(x)\in \wpwedge(C)$, so that $\fwedge(x) \in f(x)$. 
If $x,y\in C$ and $x\leq y$ then, by hypothesis, $f(x) \preceq_S f(y)$, therefore, since $\fwedge(y) \in f(y)$, there exists some
$z\in f(x)$ such that $z\leq \fwedge(y)$, and, in turn, $\fwedge(x) \leq z \leq \fwedge (y)$. Hence, since $\fwedge$ is a monotone function
on a complete lattice, 
by Tarski's theorem, its least fixed point $\lfp(\fwedge)\in C$ exists. Furthermore, by the constructive version
of Tarski's theorem \cite[Theorem~5.1]{cc79pjm}, $\lfp(\fwedge) = \bigvee_{\alpha\in \bO} \fwedge^\alpha (\bot)$. 
We have that $\lfp(\fwedge) = \fwedge(\lfp(\fwedge)) \in f(\lfp(\fwedge))$, hence $\lfp(\fwedge)\in \Fix(f)$. Consider any $z\in \Fix(f)$. 
We prove by transfinite induction that for any $\alpha \in \bO$,  $\fwedge^\alpha (\bot)\leq z$. 
If $\alpha =0$ then $\fwedge^0 (\bot)=\bot \leq z$. If $\alpha = \beta+1$ then 
$\fwedge^\alpha (\bot)= \fwedge (\fwedge^\beta (\bot))$, and, 
since, by inductive hypothesis, $\fwedge^\beta (\bot)\leq z$, then, by monotonicity of $\fwedge$, 
$\fwedge (\fwedge^\beta (\bot)) \leq \fwedge (z) = \wedge f(z) \leq z$. 
If $\alpha = \vee \{\beta\in \bO~|~\beta<\alpha\}$ is a limit ordinal then $\fwedge^\alpha (\bot)= \bigvee_{\beta <\alpha} \fwedge^\beta (\bot)$; since,
by inductive hypothesis, $\fwedge^\beta (\bot)\leq z$ for any $\beta < \alpha$, we obtain that  $\fwedge^\alpha (\bot)\leq z$.
This therefore shows that $f$ has the least fixed point $\lfp(f)=\lfp(\fwedge)$. 
\end{proof}

By duality, as consequences of the above result, we obtain the following characterizations, 
where point~(3) coincides with Zhou's theorem (see \cite[Theorem~1]{zhou94} and
\cite[Proposition 3.15]{straccia2008}), which is used for showing that pure Nash equilibria of a supermodular
game form a complete lattice. 

\begin{corollary}\label{coro4}\ \\ \
{\rm (1)} If $f:C\ra \wpvee(C)$ is $H$-monotone then $f$ has the greatest fixed point $\gfp(f)
= \bigwedge_{\alpha \in \bO} \fvee^\alpha (\top)$.\\ 
{\rm (2)} If $f:C\ra \wp^\diamond(C)$ is $\EM$-monotone then $f$ has the least and greatest fixed points, where
$\lfp(f)= \bigvee_{\alpha \in \bO} \fwedge^\alpha (\bot)$ and 
$\gfp(f)
= \bigwedge_{\alpha \in \bO} \fvee^\alpha (\top)$.\\
{\rm (3)} If $f:C\ra \SL(C)$ is $\EM$-monotone then $\tuple{\Fix(f),\leq}$ is a complete lattice.  \\
{\rm (4)} If $f,g:C\ra \SL(C)$ are $\EM$-monotone and, for any $c\in C$, $f(c) \preceq_{\EM} g(c)$ 
then $\Fix(f) \preceq_{\EM} \Fix(g)$.
\end{corollary}
\begin{proof}
Let us prove point~(4). By Point~(3), both $\Fix(f)$ and $\Fix(g)$ are complete lattices for $\leq$. Thus, 
$\Fix(f) \preceq_{\EM} \Fix(g)$ holds iff $\wedge \Fix(f) =\lfp(f) \leq \lfp(g)=\wedge \Fix(g)$ and
$\vee \Fix(f) =\gfp(f) \leq \gfp(g)=\vee \Fix(g)$. Moreover, since, for any $c\in C$, $f(c) \preceq_{\EM} g(c)$, 
we also have that $\fwedge(c) = \wedge f(c) \leq \wedge f(c) = \gwedge(c)$, thus, as a consequence, 
$\lfp(\fwedge) \leq \lfp(\gwedge)$. 
The proof of Lemma~\ref{lemmalfp} shows
that $\lfp(f)= \lfp(\fwedge)$ and $\lfp(g)= \lfp(\gwedge)$, so that we obtain $\lfp(f)\leq \lfp(g)$. The proof for 
$\gfp(f) \leq \gfp(g)$ is dual. 
\end{proof}

\subsection{Concretization-based Approximations}\label{concr}
As discussed in \cite{cc92}, a minimal requirement for defining an abstract domain consists in specifying the meaning of its
abstract values through a concretization map. 
Let $\tuple{A,\leq_A}$ be an abstraction of a concrete domain $C$ specified by a monotone concretization map $\gamma:A \ra C$.
Let us observe that the powerset lifting 
$\gamma^s:\wp(A) \ra \wp(C)$ is $S$-monotone,
meaning that if $Y_1 \preceq_S Y_2$ then $\gamma^s(Y_1) \preceq_S \gamma^s(Y_2)$: if $\gamma(y_2)\in
\gamma^s(Y_2)$ then there exists $y_1\in Y_1$ such that $y_1 \leq_A y_2$, so that $\gamma(y_1)\in
\gamma^s(Y_1)$ and $\gamma(y_1)\leq_C \gamma(y_2)$, i.e., $\gamma^s(Y_1) \preceq_S \gamma^s(Y_2)$. Analogously, $\gamma^s$ is $H$- and $\EM$-monotone. 
Consider a concrete $S$-monotone multivalued function $f:C\ra \wpwedge(C)$, whose least fixed point exists by Lemma~\ref{lemmalfp}.

\begin{definition}[\textbf{Correct Approximation of Multivalued Functions}]\rm 
An abstract multivalued function $f^\sharp:A\ra \wp(A)$ over $A$ is a \emph{$S$-correct approximation} of $f$ when:
\begin{itemize}
\item[(1)] $f^\sharp: A \ra \wpwedge(A)$ and $f^\sharp$ is $S$-monotone \qquad (fixed point condition)  
\item[(2)] for any $a\in A$, $f(\gamma(a)) \preceq_S \gamma^s (f^\sharp (a))$ \qquad\hspace{3pt} (soundness condition)
\end{itemize}
$H$- and $\EM$-correct approximations are defined by replacing in this definition 
$S$- with, respectively, $H$- and $\EM$-, and $\wpwedge$ with, respectively, $\wpvee$ and $\wp^\diamond$. \qed
\end{definition}

Let us point out that the soundness condition~(2) is the standard correctness requirement used in abstract interpretation,
as recalled in Section~\ref{ai-sec}. 
The difference here is
that $C_2$ and $A_2$ are mere preorders rather than partial orders. However, this is enough for guaranteeing a correct approximation
of least fixed points. 

\begin{theorem}[\textbf{Correct Least Fixed Point Approximation}]\label{lfp-approx}
If $f^\sharp$ is a $S$-correct approximation of $f$ then $\lfp(f) \leq_C \gamma(\lfp(f^\sharp))$.
\end{theorem}
\begin{proof}
Let us consider $\fwedge:C\ra C$ and $\fwedge^\sharp:A\ra A$.  
By Lemma~\ref{lemmalfp}, $\lfp(f) = \lfp(\fwedge)$ and 
$\lfp(f^\sharp) = \lfp(\fwedge^\sharp)$. Let us check that $\fwedge^\sharp$ is a standard correct approximation of $\fwedge$. 
For any $a\in A$, $\gamma (\fwedge^\sharp (a))
\in \gamma^s (f^\sharp (a))$, hence, since $f(\gamma(a)) \preceq_S \gamma^s (f^\sharp (a))$, we have that there exists some $z\in f(\gamma(a))$
such that $z \leq \gamma(\fwedge^\sharp(a))$, so that $\fwedge (\gamma(a)=\wedge f(\gamma(a)) \leq z \leq \gamma(\fwedge^\sharp(a))$. 
Hence, by the concretization-based fixed point transfer (see \cite[Theorem~2.2.4]{minephd}), it turns out that
$\lfp(\fwedge) \leq_C \gamma(\lfp(\fwedge^\sharp))$, therefore showing that  $\lfp(f) \leq \gamma(\lfp(f^\sharp))$.
\end{proof}

Dual results hold for $H$- and $\EM$-correct approximations.  

\begin{corollary}\label{coro-main}\ \\ \
{\rm (1)} If $f^\sharp$ is a $H$-correct approximation of $f$ then $\gfp(f) \leq_C \gamma(\gfp(f^\sharp))$.\\
{\rm (2)} If $f^\sharp$ is a $\EM$-correct approximation of $f$ then
$\Fix(f) \preceq_{\EM} \gamma^s(\Fix(f^\sharp))$, in particular, $\lfp(f) \leq_C \gamma(\lfp(f^\sharp))$ and 
$\gfp(f) \leq_C \gamma(\gfp(f^\sharp))$.
\end{corollary}
\begin{proof}
By duality from Theorem~\ref{lfp-approx}. In particular, point (2) follows because, by Corollary~\ref{coro4}, 
$\Fix(f)\in \wp^\diamond(C)$, $\Fix(f^\sharp)\in\wp^\diamond(A)$ and therefore $\gamma^s(\Fix(f^\sharp))\in \wp^\diamond(C)$, 
so that $\Fix(f) \preceq_{\EM} \gamma^s(\Fix(f^\sharp))$ iff $\lfp(f) \leq \gamma(\lfp(f^\sharp))$ and
$\gfp(f) \leq \gamma(\gfp(f^\sharp))$.
\end{proof}

The approximation of least/greatest fixed points of multivalued functions can also be easily given for an abstraction map 
$\alpha:C\ra A$. In this case, a $S$-monotone map $f^\sharp:A \ra \wpwedge(A)$ is a correct approximation of a concrete $S$-monotone map
$f:C\ra \wpwedge(C)$ when, for any $c\in C$, $\alpha^s(f(c)) \preceq_S f^\sharp (\alpha(c))$, where $\alpha^s:\wp(C)\ra \wp(A)$.
Here, fixed point approximation states that $\alpha(\lfp(f)) \leq_A
\lfp(f^\sharp)$.

\subsection{Galois Connection-based Approximations}
Let us now consider the ideal case of abstract interpretation where the
best
approximations in an abstract domain $A$
of concrete objects always exist, that is, 
$A$ is specified by a GC $(\alpha,C,A,\gamma)$.
However, 
recall that here $\tuple{\wpwedge(C),\preceq_S}$ and $\tuple{\wpwedge(A),\preceq_S}$ are mere preorders, and not posets. 
Then, given two preorders $\tuple{X,\preceq_X}$ and $\tuple{Y,\preceq_Y}$, we say that 
two functions $\beta:X\ra Y$ and $\delta:Y\ra X$ specify
a preorder-GC $(\beta,X,Y,\delta)$ when $\delta$ and $\beta$ are monotone (meaning, e.g.\ for $\beta$, 
that $x\preceq_X x'\,\Ra\, \beta(x) \preceq_Y \beta(x')$) and the equivalence $\beta(x) \preceq_Y y \,\Lra\, x \preceq_X \delta(y)$ holds.  
As expected, 
it turns out that GCs induce preorder-GCs for Smyth, Hoare and Egli-Milner preorders. 

\begin{lemma}\label{lemmaGC}
Let $(\alpha,C,A,\gamma)$ be a Galois connection. 
Then,
$\big(\alpha^s, \tuple{\wpwedge(C),\preceq_S}, \tuple{\wpwedge(A),\preceq_S},\gamma^s\big)$, $\big(\alpha^s,$ ${\tuple{\wpvee(C),\preceq_H},} \tuple{\wpvee(A),\preceq_H},\gamma^s\big)$, 
and $\big(\alpha^s, \tuple{\wp^\diamond(C),\preceq_{\EM}}, \tuple{\wp^\diamond(A),\preceq_{\EM}},\gamma^s\big)$ are preorder-Galois connections. 
\end{lemma}
\begin{proof}
Let us check that $\alpha^s$ is $S$-monotone: if $X\preceq_S Y$ and $\alpha(y)\in \alpha^s(Y)$ then there exists $x\in X$ such that 
$x\leq_C y$, so that, by monotonicity of $\alpha$, $\alpha(x) \leq_A \alpha(y)$, and therefore $\alpha^s(X)\preceq_S \alpha^s(Y)$. 
Analogously, $\gamma^s$ is $S$-monotone. Let us check that $\alpha^s(X) \preceq_S Y \,\Ra \, X \preceq_S \gamma^s(Y)$: 
if $\gamma(y)\in \gamma^s(Y)$ then there exists $\alpha(x)\in \alpha^s(X)$ such that $\alpha(x) \leq_A y$, and, since 
$(\alpha,C,A,\gamma)$ is a GC, this implies that $x\leq_C \gamma(y)$, so that $X \preceq_S \gamma^s(Y)$. Analogously, 
it turns out that $X \preceq_S \gamma^s(Y)\, \Ra\, \alpha^s(X) \preceq_S Y$. Hence, this shows that 
$\big(\alpha^s, \tuple{\wpwedge(C),\preceq_S}, \tuple{\wpwedge(A),\preceq_S},\gamma^s\big)$ is a preorder-GC. 
The proofs for Hoare and Egli-Milner preorders are analogous. 
\end{proof}

The ideal Galois connection-based framework allows us to define best correct approximations of multivalued functions. 
If $f:C\ra \wp(C)$ and $(\alpha,C,A,\gamma)$ is a GC then its \emph{best correct approximation} on the abstract domain $A$ is
the multifunction $f^A: A \ra \wp(A)$ defined as follows: $f^A(a) \ud \alpha^s(f(\gamma(a)))$. 
In particular, if $f:C\ra \wpwedge(C)$ is $S$-monotone then $f^A:A\ra \wpwedge(A)$ turns out to be $S$-monotone.
Analogously for Hoare and Egli-Milner preorders. 
Similarly to standard abstract
interpretation \cite{CC79}, it turns out that $f^A$ is the best among the $S$-correct approximations of $f$,
as formalized by the following result. 
\begin{lemma}
A $S$-monotone
correspondence
$f^\sharp:A \ra \wpwedge(A)$ is a $S$-correct approximation of $f$ iff for any $a\in A$, $f^A(a) \preceq_S f^\sharp(a)$. 
Also, analogous characterizations hold for $H$- and $\EM$-correct approximations. 
\end{lemma}
\begin{proof}
An easy consequence of Lemma~\ref{lemmaGC}, since for any $a\in A$, 
$f^A(a)=\alpha^s(f(\gamma(a)) \preceq_S f^\sharp(a)$ iff for any $a\in A$, $f(\gamma(a)) \preceq_S \gamma^s(f^\sharp(a))$. 
\end{proof}

Hence, it turns out that the fixed point approximations given by 
Theorem~\ref{lfp-approx}  and Corollary~\ref{coro-main} apply to the best correct approximations $f^A$. 

\paragraph{Completeness.} In abstract interpretation, completeness \cite{CC79,grs00} formalizes an ideal situation where the abstract function $f^\sharp$ on $A$
is capable of not losing information w.r.t.\ the abstraction in $A$ of the concrete function $f$, that is, the equality
$\alpha(f(c))=f^\sharp(\alpha(c))$ always holds. As a key consequence, completeness lifts to fixed points, meaning that $\alpha(\lfp(f)) =
\lfp(f^\sharp)$ holds. Let us show that this also holds for multivalued functions.  
An abstract $S$-monotone function $f^\sharp:A\ra \wpwedge(A)$ is a \emph{complete approximation} of a $S$-monotone function 
$f:C\ra \wpwedge(C)$ when for any $c\in C$, $\alpha^s(f(c)) = f^\sharp(\alpha(c))$. 

\begin{lemma}[\textbf{Complete Least Fixed Point Approximation}]
If $f^\sharp$ is a complete approximation of $f$ then $\alpha(\lfp(f)) = \lfp(f^\sharp)$.
\end{lemma}
\begin{proof}
By Lemma~\ref{lemmalfp}, $\lfp(f) = \lfp(\fwedge)$ and 
$\lfp(f^\sharp) = \lfp(\fwedge^\sharp)$. Since $\fwedge (c) \in f(c)$, we have that $\alpha(\fwedge (c)) \in \alpha^s (f(c))$, so
that $ \alpha(\fwedge(c))=\wedge \alpha^s(f(c))$. By hypothesis, $\wedge \alpha^s(f(c)) = \wedge f^\sharp (\alpha(c)) = \fwedge^\sharp(\alpha(c))$, 
so that $ \alpha \circ \fwedge = \fwedge^\sharp \circ \alpha$ holds. Thus, by complete fixed point transfer \cite[Theorem~7.1.0.4]{CC79}, 
$\alpha(\lfp(\fwedge)) = \lfp(\fwedge^\sharp)$.
\end{proof}

\subsection{Approximations of Best Response Correspondences}\label{abrc}

The above abstract interpretation-based approach for
multivalued functions can be applied to (quasi)super\-modular games by approximating their best response correspondences. 
In particular, one can abstract both the $i$-th best response correspondences $B_i:S_{-i}\ra \SL(S_i)$ and
the overall best response $B:S\ra \SL(S)$.  

\begin{example}\label{ex2}\rm
Let us consider the game $\Gamma$ in Example~\ref{ex1} and the abstraction $A$ of its strategy space $C\times C$
defined in Example~\ref{ex-comp}. 
Then, one can define the best correct approximation $B^A$ in $A$ of the best response function $B:C\times C \ra \SL(C\times C)$, that is,
$B^A:A \ra \wp(A)$ is defined as $B^A(a) \ud \alpha^s(B(\gamma(a))=\alpha^s(B(a))=\{\alpha(s_1,s_2)\in A~|~ (s_1,s_2) \in B(a)\}$. We therefore have that:
\begin{align*}
&B^A(2,2) = \alpha^s(\{(2,3)\})=\{(3,4)\},~~B^A(3,4) = \alpha^s(\{(2,3),(5,3)\})=\{(3,4),(6,6)\},\\
&B^A(4,4)= \alpha^s(\{(2,4),(5,4)\}) = \{(3,4),(6,6)\},~~B^A(3,5) = \alpha^s(\{(5,3)\}) = \{(6,6)\},\\
&B^A(4,5) = \alpha^s(\{((5,4)\}) = \{(6,6)\},~~B^A(6,6) = \alpha^s(\{(5,4),(6,4)\}) = \{(6,6)\}.
\end{align*}
Hence, $\Fix(B^A)=\{(3,4),(6,6)\}$. 
Therefore, by Theorem~\ref{lfp-approx} and Corollary~\ref{coro-main}, here we have that $\lne(\Gamma)=\lfp(B) = (2,3) \leq (3,4) = \lfp(B^A)$ and $\gne(\Gamma) = \gfp(B) = (5,4) \leq (6,6) =\gfp(B^A)$. 
\qed
\end{example}


\section{Games with Abstract Strategy Spaces}

Let us consider a game $\Gamma = \langle S_i,u_i\rangle_{i=1}^n$ and a corresponding  
family $\cG = (\alpha_i,S_i ,A_i,\gamma_i)_{i=1}^n$ of GCs of the strategy spaces $S_i$. 
By Lemma~\ref{prod2}, $(\alpha,\times_{i=1}^n S_i,\times_{i=1}^n A_i,\gamma)$ specifies a nonrelational product abstraction 
of the whole strategy space $\times_{i=1}^n S_i$. 
We define the $i$-th utility function $u_i^\cG: \times_{i=1}^n A_i \ra \bR^{N_i}$ on the abstract strategy space 
$\times_{i=1}^n A_i$ simply by restricting $u_i$ on $\gamma(\times_{i=1}^n A_i)$ as follows: $u_i^\cG(a) \ud u_i(\gamma(a))$. 
We point out that this definition is a form of generalization of the restricted games considered by Echenique~\cite[Section~2.3]{ech2007}.

\begin{lemma}\label{lemma-sm2}
If  $u_i(\cdot,s_{-i})$ is (quasi)supermodular and all the GCs in $\cG$ are finitely disjunctive then 
$u_i^\cG(\cdot, a_{-i}):A_i \ra \bR^{N_i}$ is (quasi)super\-modular. Also, 
if $u_i(s_i, \cdot)$ is monotone then $u_i^\cG(a_i, \cdot):A_{-i}\ra \bR^{N_i}$ is monotone. 
\end{lemma}
\begin{proof}
Let us check that $u_i^\cG(\cdot, a_{-i})$ is supermodular:
\begin{align*}
u_i^\cG (a_i \vee_{A_i} b_i,a_{-i}) + u_i^\cG(a_i \wedge_{A_i} b_i,a_{-i}) &= \qquad \text{[by definition]}\\
u_i (\gamma_i(a_i \vee_{A_i} b_i),\gamma_{-i}(a_{-i})) + 
u_i (\gamma_i(a_i  \wedge_{A_i} b_i),\gamma_{-i}(a_{-i})) &= \qquad \text{[$\cG$ are finitely disjunctive GCs]}\\
u_i (\gamma_i(a_i) \vee_{i} \gamma_i(b_i),\gamma_{-i}(a_{-i})) +  u_i (\gamma_i(a_i)  \wedge_{i} \gamma_i(b_i),\gamma_{-i}(a_{-i}))
&\geq \qquad \text{[by supermodularity of $u_i$]}\\
u_i (\gamma_i(a_i),\gamma_{-i}(a_{-i})) +  u_i (\gamma_i(b_i),\gamma_{-i}(a_{-i}))
&= \qquad \text{[by definition]}\\
u_i^\cG (a_i,a_{-i}) + u_i^\cG(b_i,a_{-i})
\end{align*}
The proof of quasisupermodularity is analogous. Let us also check that $u_i^\cG(a_i, \cdot)$ is monotone. Consider
$a_{-i} \leq b_{-i}$, so that, by monotonicity of $\gamma_{-i}$, we have that $\gamma_{-i}(a_{-i}) \leq \gamma_{-i}(b_{-i})$. Hence,
by monotonicity of $u_i (\gamma_i(a_i), \cdot)$, we obtain:
$u_i^\cG (a_i, a_{-i}) =
u_i (\gamma_i(a_i), \gamma_{-i}(a_{-i}))\leq u_i (\gamma_i(a_i), \gamma_{-i}(b_{-i})) = u_i^\cG (a_i, b_{-i})$.
\end{proof}

Let us also 
observe that if $u_i(s_i,s_{-i})$ has increasing differences (the single crossing property), $X\subseteq \times_{i=1}^n S_i$ is any subset 
of the strategy space and
$u_{i_{{\!/\!X}}}:X\ra \bR^{N_i}$ is the mere restriction of $u_i$ to the subset $X$ then $u_{i_{{\!/\!X}}}$ still has increasing differences (the single crossing property). 
Hence, in particular, this holds for $u_i^\cG:\times_{i=1}^n A_i \ra \bR$. As a consequence of this and of Lemma~\ref{lemma-sm2}, we obtain
the following abstract (quasi)supermodular games. 
      
\begin{corollary}\label{coro-sm2}
If $\Gamma=\langle S_i,u_i\rangle_{i=1}^n$ is a (quasi)supermodular game and $\cG = (\alpha_i,S_i,A_i,\gamma_i)_{i=1}^n$ is a family of
finitely disjunctive GCs then  $\Gamma^\cG \ud \langle A_i,u_i^\cG\rangle_{i=1}^n$ is a (quasi)supermodular game.
\end{corollary}

Let us see an array of examples of abstract games. 

\begin{example}\label{ex3}\rm
Consider the game $\Gamma$ in Example~\ref{ex1} and the product abstraction 
$A_1\times A_2\in \Abs(S_1\times S_2)$ defined in Example~\ref{ex-comp}. 
The restricted game $\Gamma^\sharp$ of Lemma~\ref{lemma-sm2} on the abstract strategy space $\{3,5,6\}\times \{2,6\}$ is therefore
specified by the following payoff matrix:

\begin{center}
\renewcommand{\arraystretch}{1.2}
\begin{tabular}{c|C{25pt}|C{25pt}|}
\multicolumn{1}{c}{}
& \multicolumn{1}{c}{\textbf{2}} &  \multicolumn{1}{c}{\textbf{6}}\\
\cline{2-3}
\textbf{6} &-1,\,-1 &6,\,5 \\
\cline{2-3}
\textbf{5} & 0,\,2 & 6,\,5 \\
\cline{2-3}
\textbf{3} & 2,\,4 & 3,\,2 \\
\cline{2-3}
\end{tabular}
\end{center}

\medskip
\noindent
Since both $A_1$ and $A_2$ are trivially disjunctive abstractions, 
by Corollary~\ref{coro-sm2}, it turns out that $\Gamma^\sharp$ is supermodular. The best response correspondences $B_i^\sharp:A_{-i}\ra \SL(A_i)$ 
for the supermodular game $\Gamma^\sharp$ are therefore as follows:
\[
\begin{array}{llllll}
B_1^\sharp(2)=\{3\},& B_1^\sharp(6)=\{5,6\},& 
B_2^\sharp(3)=\{2\}; & B_2^\sharp(5)=\{6\},& B_2^\sharp (6)=\{6\}.& 
\end{array}
\]
We observe that $B_2^\sharp$ is not a $S$-correct approximation of $B_2$ because: 
$B_2(3)=\{3\} \not{\!\,\!\!\preceq_S}\: \{2\} = B_2^\sharp(3)$. Indeed, it turns out that 
$\Eq(\Gamma^\sharp)=\{(3,2), (5,6), (6,6)\}$, so that 
$\lne(\Gamma)=(2,3) \not\leq (3,2) = \lne(\Gamma^\sharp)$.
Thus, in this case, the solutions of the abstract game $\Gamma^\sharp$ do not correctly approximate the solutions of $\Gamma$. 

\noindent
Instead, following Section~\ref{abrc} and analogously to Example~\ref{ex2}, 
one can define the best correct approximation $B^A:A\ra \SL(A)$ in $A\ud A_1 \times A_2$ of the best response correspondence 
$B$ of $\Gamma$,
that is, $B^A(a_1,a_2) = \{(\alpha_1(s_1),\alpha_2(s_2)) \in A~|~ (s_1,s_2) \in B(a_1,a_2)\}$ acts as follows:
\begin{align*}
&B^A(3,2) = \{(3,6)\},~~B^A(3,6) = \{(5,6),(6,6)\},~~B^A(5,2) = \{(3,6)\},\\
&B^A(5,6) = \{(5,6),(6,6)\},~~B^A(6,2) = \{(3,6)\},~~B^A(6,6) = \{(5,6),(6,6)\}.
\end{align*}
Hence, $\Fix(B^A)=\{(5,6),(6,6)\}$, so that 
$\lne(\Gamma)=\lfp(B) = (2,3) \leq (5,6) = \lfp(B^A)$ and $\gne(\Gamma) = \gfp(B) = (5,4) \leq (6,6) =\gfp(B^A)$. \qed
\end{example}

\begin{example}\label{ex4}\rm 
In Example~\ref{ex3}, let us consider the abstraction $A_2=\{4,6\}\in \Abs(S_2)$, so that the supermodular game
$\Gamma^\sharp$ is given by the following payoff matrix: 
\begin{center}
\renewcommand{\arraystretch}{1.2}
\begin{tabular}{c|C{25pt}|C{25pt}|}
\multicolumn{1}{c}{}
& \multicolumn{1}{c}{\textbf{4}} &  \multicolumn{1}{c}{\textbf{6}}\\
\cline{2-3}
\textbf{6} &5,\,6 &6,\,5 \\
\cline{2-3}
\textbf{5} & 6,\,6 & 6,\,5 \\
\cline{2-3}
\textbf{3} & 4,\,5 & 3,\,2 \\
\cline{2-3}
\end{tabular}
\end{center}
%
\noindent
while the best response correspondences $B_i^\sharp$ become: 
\[
\begin{array}{llllll}
B_1^\sharp(4)=\{5\},& B_1^\sharp(6)=\{5,6\},& B_2^\sharp(3)=\{4\}; &
B_2^\sharp(5)=\{4\},& B_2^\sharp (6)=\{4\}.& 
\end{array}
\]
Thus, here we have that $\Eq(\Gamma^\sharp)= \{(5,4)\}$. 
In this case, it turns out that $B_i^\sharp$ is a $\EM$-correct approximation of $B_i$, so that, by Corollary~\ref{coro-main}~(2),
$\Eq(\Gamma) = \Fix(B)=\{(2,3),(5,4)\} \preceq_{\EM} \{(5,4)\}=\Fix(B^\sharp)=\Eq(\Gamma^\sharp)$ holds.
\qed
\end{example}

\begin{example}\label{ex5}\rm
Here, we consider the disjunctive abstractions $A_1 =\{4,5,6\}\in \Abs(S_1)$ and 
$A_2=\{3,4,5,6\}\in \Abs(S_2)$.   
In this case, we have the following supermodular abstract game $\Gamma^\sharp$ over $A_1\times A_2$:

\begin{center}
\renewcommand{\arraystretch}{1.2}
\begin{tabular}{c|C{25pt}|C{25pt}|C{25pt}|C{25pt}|}
\multicolumn{1}{c}{}
& \multicolumn{1}{c}{\textbf{3}} &  \multicolumn{1}{c}{\textbf{4}} &  \multicolumn{1}{c}{\textbf{5}} &  \multicolumn{1}{c}{\textbf{6}}\\
\cline{2-5}
\textbf{6} &2,\,4 &5,\,6 &6,\,5 &6,\,5\\
\cline{2-5}
\textbf{5} & 3,\,4 & 6,\,6 & 7,\,5 &6,\,5 \\
\cline{2-5}
\textbf{4} & 3,\,5 & 5,\,6 & 5,\,5 &4,\,4\\
\cline{2-5}
\end{tabular}
\end{center}

\medskip
\noindent
where the best response functions $B_i^\sharp$ are therefore as follows:
\[
\begin{array}{lllll}
B_1^\sharp(3)=\{4,5\},& B_1^\sharp(4)=\{5\},& B_1^\sharp(5)=\{5\},& B_1^\sharp(6)=\{5,6\}; \\
B_2^\sharp(4)=\{4\},& B_2^\sharp(5)=\{4\},& B_2^\sharp (6)=\{4\}.& 
\end{array}
\]
Here, it turns out that $B_i^\sharp$ is a $\EM$-correct approximation of $B_i$, so that 
the abstract best response $B^\sharp:A_1\times A_2 \ra \SL(A_1\times A_2)$ is a $\EM$-correct approximation of $B$. 
Then, by Corollary~\ref{coro-main}~(2), we have that
$\Eq(\Gamma)= \Fix(B)=\{(2,3),(5,4)\} \preceq_{\EM} \{(5,4)\}=\Fix(B^\sharp)=\Eq(\Gamma^\sharp)$.
\qed
\end{example}

Thus, for the concrete supermodular game $\Gamma$ of Example~\ref{ex1}, 
while the abstract games of Examples~\ref{ex4} and~\ref{ex5} can be viewed as correct approximations of $\Gamma$, this 
instead does not hold for the abstract game in Example~\ref{ex3}. The following results provide conditions 
that justify these different behaviors.

\begin{theorem}[\textbf{Correctness of Games with Abstract Strategy Spaces}]\label{theo-con}
Let $\cG=(\alpha_i,S_i,A_i,\gamma_i)_{i=1}^n$ be a family of finitely disjunctive GIs, $S=\times_{i=1}^n S_i$, $A=\times_{i=1}^n A_i$ 
and $(\alpha,S,A,\gamma)$ be
the nonrelational product composition of $\cG$. Let $\Gamma=\langle S_i,u_i\rangle_{i=1}^n$ be a (quasi)supermodular game, with best response $B$, 
and $\Gamma^\cG=\langle A_i, u_i^\cG\rangle_{i=1}^n$ be the corresponding abstract
(quasi)supermodular game, with best response $B^\cG$. 
Assume that for any $a\in A$, 
$\bigvee_S B(\gamma(a)) \vee_S \gamma(\bigwedge_A B^\cG(a)) \in \gamma(A)$. Then,
$\Eq(\Gamma) \preceq_{\EM} \gamma^s(\Eq(\Gamma^\cG))$ and, in particular, 
$\lne(\Gamma) \leq \gamma^s(\lne(\Gamma^\cG))$ and $\gne(\Gamma) \leq \gamma^s(\gne(\Gamma^\cG))$.
\end{theorem}
\begin{proof}
We have that $\Eq(\Gamma)=\Fix(B)$ and
$\Eq(\Gamma^\cG)=\Fix(B^\cG)$, where $B:S\ra \wp^\diamond(S)$ and $B^\cG: A \ra \wp^\diamond(A)$ are $\EM$-monotone. Thus,
by Corollary~\ref{coro-main}~(2), in order to prove that $\Eq(\Gamma) \preceq_{\EM} \gamma^s(\Eq(\Gamma^\cG))$ it is enough to prove
that for any $a\in A$, $B(\gamma(a)) \preceq_{\EM} \gamma^s(B^\cG(a))$. Let $h\ud \bigvee_S B(\gamma(a))\in S$, so that $h\in B(\gamma(a))$, 
and $k\ud \bigwedge_A B^\cG(a)\in A$, so that, by Corollary~\ref{coro-sm2}, $k\in B^\cG(a)$. By hypothesis, we have that $h\vee_S \gamma(k) \in \gamma(A)$. 
Let us consider some $i\in [1,n]$. Therefore,   
$h_i \vee_i \gamma_i(k_i) \in \gamma_i(A_i)$, that is, $h_i \vee_i \gamma_i(k_i) = \gamma_i (b_i)$, for some $b_i\in A_i$. 
Hence, since $k_i \in B^\cG_i (a_{-i})$, we have that 
$$u_i(h_i \vee_i \gamma_i(k_i), \gamma_{-i}(a_{-i})) = u_i(\gamma_i(b_i), \gamma_{-i}(a_{-i})) = u_i^\cG(b_i, a_{-i}) \leq u_i^\cG(k_i, a_{-i}) = 
u_i(\gamma_i(k_i), \gamma_{-i}(a_{-i})).$$
On the other hand, since $h_i \in B_i(\gamma(a)_{-i}) = B_i(\gamma_{-i}(a_{-i}))$, we have that
$u_i(h_i \wedge_i \gamma_i(k_i), \gamma_{-i}(a_{-i})) \leq u_i (h_i, \gamma_{-i}(a_{-i}))$.
Furthermore, since $u_i$ is supermodular, we also have that 
$$u_i(h_i \wedge_i \gamma_i(k_i), \gamma_{-i}(a_{-i})) + u_i(h_i \vee_i \gamma_i(k_i), \gamma_{-i}(a_{-i})) \geq 
u_i(h_i,\gamma_{-i}(a_{-i})) + u_i(\gamma_i(k_i), \gamma_{-i}(a_{-i})).$$
We therefore obtain:
\begin{align*}
u_i(h_i,\gamma_{-i}(a_{-i})) + u_i(\gamma_i(k_i), \gamma_{-i}(a_{-i})) &\geq
u_i(h_i \wedge_i \gamma_i(k_i), \gamma_{-i}(a_{-i})) + u_i(h_i \vee_i \gamma_i(k_i), \gamma_{-i}(a_{-i}))\\
&\geq 
u_i(h_i,\gamma_{-i}(a_{-i})) + u_i(\gamma_i(k_i), \gamma_{-i}(a_{-i}))
\end{align*}
so that $$u_i(h_i,\gamma_{-i}(a_{-i})) + u_i(\gamma_i(k_i), \gamma_{-i}(a_{-i})) =
u_i(h_i \wedge_i \gamma_i(k_i), \gamma_{-i}(a_{-i})) + u_i(h_i \vee_i \gamma_i(k_i), \gamma_{-i}(a_{-i}))$$
and, in turn, $u_i(h_i \wedge_i \gamma_i(k_i), \gamma_{-i}(a_{-i})) = u_i(h_i,\gamma_{-i}(a_{-i}))$ and 
$u_i^\cG(b_i, a_{-i}) = u_i(h_i \vee_i \gamma_i(k_i), \gamma_{-i}(a_{-i})) = u_i(\gamma_i(k_i), \gamma_{-i}(a_{-i}))=u_i^\cG(k_i, a_{-i})$. Thus,  
$h_i \wedge_i \gamma_i(k_i) \in B_i (\gamma_{-i}(a_{-i}))$  and $h_i \vee_i \gamma_i(k_i)\in \gamma_i(B_i^\cG(a_{-i}))$. Therefore,
it turns out that $h \wedge \gamma(k) \in B(\gamma(a))$ and $h \vee \gamma(k) \in \gamma^s (B^\cG(a))$.  
Hence, if $s\in B(\gamma(a))$ then $s \leq h \leq h \vee \gamma(k) \in \gamma^s (B^\cG(a))$, while if $t\in \gamma^s (B^\cG(a))$
then $t=\gamma^s(d)$, for some $d\in B^\cG(a)$, so that $k\leq_A d$ and, in turn, 
$t=\gamma^s(d) \geq \gamma(k) \geq  h \wedge \gamma(k) \in B(\gamma(a))$,
thus showing that $B(\gamma(a)) \preceq_{\EM} \gamma^s(B^\cG(a))$. The proof for quasisupermodular games is analogous. 
\end{proof}

As a consequence of the above result, we obtain a generalization of \cite[Lemma~4]{ech2007}, 
which is the basis for designing the efficient algorithm in \cite[Section~4]{ech2007} that computes
all the Nash equilibria in a finite game with strategic complementarities. 
A GC $(\alpha,C,A,\gamma)$ is called a \emph{principal filter GC} when the image 
$\gamma(A)$ is the principal filter at $\gamma(\bot_A)$, that is, 
$\gamma(A)= \{c\in C~|~ \gamma(\bot_A) \leq c\}$. 

\begin{corollary}\label{coro-eche}
Let $\cG=(\alpha_i,S_i,A_i,\gamma_i)_{i=1}^n$ be principal filter GCs. Then,  $\Eq(\Gamma) \preceq_{\EM} \gamma^s(\Eq(\Gamma^\cG))$. 
\end{corollary}
\begin{proof}
Observe that the product $(\alpha,\times_{i=1}^n S_i ,\times_{i=1}^n A_i,\gamma)$ is a 
principal filter GC. Then, this is
a straight consequence of Theorem~\ref{theo-con}, since 
$\bigvee_S B(\gamma(a)) \vee_S \gamma(\bigwedge_A B^\cG(a)) \geq \gamma(\bigwedge_A B^\cG(a))
\geq \gamma((\bot_{A_i})_{i=1}^n)$, so that  $\bigvee_S B(\gamma(a)) \vee_S \gamma(\bigwedge_A B^\cG(a))\in \gamma(A)$ holds. 
\end{proof}

\begin{example}\rm 
Let us consider the following finite supermodular game $\Delta$ taken from \cite[Example~8.11]{ch11}, which is an example
of the well known Bertrand oligopoly model~\cite{topkis98}. Players $i\in \{1,2,3\}$ stand for firms which sell substitute products $p_i$
(e.g., a can of beer), whose feasible selling prices (e.g., in euros) $s_i$ range in $S_i\ud [a,b]$, where the smallest price shift is 5 cents.
The payoff function $u_i:S_1\times S_2 \times S_3 \ra \bR$ models the profit of firm $i$:
$$u_i(s_1,s_2,s_3) \ud d_i(s_1,s_2,s_3)(s_i-c_i) $$
where $d_i(s_1,s_2,s_3)$ gives the demand of $p_i$, i.e., how many units of $p_i$ the firm $i$ sells in a given time frame, while $c_i$ is the unit cost
of $p_i$ so that $(s_i-c_i)$ is the profit per unit. Following \cite[Example~8.11]{ch11}, let us assume that:
\begin{align*}
u_1(s_1,s_2,s_3) &= (370+213(s_2+s_3) + 60s_1 -230s_1^2)(s_1 - 1.10)\\
u_2(s_1,s_2,s_3) &= (360+233(s_1+s_3) + 55s_2 -220s_2^2)(s_2 - 1.20)\\
u_3(s_1,s_2,s_3) &= (375+226(s_1+s_2) + 50s_3 -200s_3^2)(s_3 - 1.25)
\end{align*}
As shown in general in \cite[Corollary~8.9]{ch11}, it turns out that each payoff function $u_i$ has increasing differences
and $u_i(s_i,\cdot)$ is monotone, so that the game $\Delta$ has the least and greatest price equilibria $\lne(\Delta)$ and $\gne(\Delta)$, 
and $\gne(\Delta)$ ($\lne(\Delta)$) provides the best (least) profits among all equilibria. It should be noted that 
\cite[Example~8.11]{ch11} considers as payoff functions the integer part of $u_i$, namely, $\lfloor u_i(s_1,s_2,s_3)\rfloor$, however we notice
that that this definition of payoff function does not have increasing differences, so that \cite[Corollary~8.9]{ch11}, which assumes
the hypothesis of increasing differences, 
cannot be applied: for example, \cite[Example~8.11]{ch11} considers $S_i = \{ x/20~|~x\in [26,42]_{\mathbb{Z}}\}$ and
with $(1.3,1.3,1.8) \leq (1.35,1.3,1.85)$, we would have that
\begin{align*}
\lfloor u_1(1.35,1.3,1.8)\rfloor - \lfloor u_1(1.3,1.3,1.8)\rfloor & = \lfloor 173.03125 \rfloor - \lfloor 143.92\rfloor = 30 >\\
\lfloor u_1(1.35,1.3,1.85)\rfloor - \lfloor u_1(1.3,1.3,1.85) \rfloor&= \lfloor 175.69375 \rfloor - \lfloor 146.05\rfloor = 29
\end{align*}

\noindent
Instead, we consider here $S_i \ud \{ x/20~|~x\in [20,46]_{\mathbb{Z}}\}$, namely the feasible prices range
from 1 to 2.3 euros with 0.05  shift.  Using the standard RT algorithm in Figure~\ref{fig-algo} 
(we made a simple C++ implementation of RT), one obtains:
$$\lne(\Delta)=(1.80,1.90,1.95)=\gne(\Delta)$$
namely, $\Delta$ admits a unique Nash equilibrium. It turns out that the algorithm
RT calculates $\lne(\Delta)$ starting from the bottom $(1.0,1.0,1.0)$ through 12 calls to 
$\bigwedge B_i(s_{-i})$, while it may output the same equilibrium as 
$\gne(\Delta)$ beginning from the top
$(2.3,2.3,2.3)$ through 9 calls to $\bigvee B_i(s_{-i})$. 

\noindent
Let us consider the following abstractions $A_i\in \Abs(S_i)$:
\[
A_1 \ud \{ x/20~|~x\in [35,38]_{\mathbb{Z}} \cup [42,46]_{\mathbb{Z}}\},\quad
A_2 \ud \{ x/20~|~x\in [36,46]_{\mathbb{Z}}\}, \quad
A_3 \ud \{ x/20~|~x\in [38,46]_{\mathbb{Z}} \}.
\]
Notice that $A_2$ and $A_3$ are principal filter abstractions, while this is not the case for $A_1$, so that
Corollary~\ref{coro-eche} cannot be applied. 
We observe that: 
\begin{align*}
\{\textstyle{\bigvee_1} B_1(a_{-1})\in S_1 ~|~a_{-1}\in A_2\times A_3\}&= \{ 36/20, 37/20, 38/20\}, \\
\{\textstyle{\bigvee_2} B_2(a_{-2})\in S_2 ~|~a_{-2}\in A_1\times A_3\}&=\{ 38/20, 39/20, 40/20\}, \\
\{\textstyle{\bigvee_3} B_3(a_{-3})\in S_3 ~|~a_{-3}\in A_1\times A_2\}&=\{ 39/20, 40/20, 41/20, 42/20\}.
\end{align*}
The hypothesis of Theorem~\ref{theo-con} is therefore satisfied, because
for any $a_{-i}\in A_{-i}$, we have that $\bigvee B_i(a_{-i}) \in A_i$. Hence, by Corollary~\ref{coro-sm2},
we consider the supermodular abstract game $\Delta^A$ on the abstract strategy spaces $A_i$. 
By exploiting the standard RT algorithm in Figure~\ref{fig-algo} for $\Delta^A$, we still obtain a unique equilibrium $\lne(\Delta^A)=(1.80,1.90,1.95)=\gne(\Delta^A)$,
so that in this case no approximation of equilibria occurs. Here,  
RT calculates $\lne(\Delta^A)$ starting from the bottom $(1.8,1.8,1.9)$ of $A_1\times A_2 \times A_3$ through 6 calls to 
$\bigwedge B_i^A(a_{-i})$ and any call $\bigwedge B_i^A(a_{-i})$ scans the smaller abstract strategy space $A_i$ instead of $S_i$. On the other hand,  
$(1.80,1.90,1.95)=\gne(\Delta)$ can be also calculated from the top
$(2.3,2.3,2.3)$ still with  9 calls to $\bigvee B_i^A(a_{-i})$, each scanning the reduced abstract strategy spaces $A_i$. 
\qed
\end{example}

\section{Games with Abstract Best Response}
In the following, we put forward a notion of abstract game where the strategy spaces are subject to a form of
partial approximation by abstract interpretation, meaning that we consider approximations of the strategy spaces
of the ``other players'' for any utility function, i.e., correct approximations of the functions $u_i(s_i, \cdot)$, for any given $s_i$. 
This approach gives rise to games having an abstract best response correspondence.  
Here, we aim at providing a systematic abstraction framework for the implicit methodology of approximate computation of 
equilibria 
considered by Carl and Heikkil{\"{a}} \cite{ch11} in their Examples~8.58, 8.63 and~8.64.

Given a game $\Gamma = \langle S_i,u_i\rangle_{i=1}^n$, we consider a  
family $\cG = (\alpha_i,S_i ,A_i,\gamma_i)_{i=1}^n$ of GCs and,
by Lemma~\ref{prod2}, their nonrelational product $(\alpha,\times_{i=1}^n S_i,\times_{i=1}^n A_i,\gamma)$, where
we denote by $\rho\ud \gamma\circ \alpha\in \uco(\times_{i=1}^n S_i)$ the corresponding closure operator
and, for any $i$, by $\rho_{-i}\in \uco(S_{-i})$ the closure operator corresponding to the $(n-i)$-th nonrelational product
$(\alpha_{-i},\times_{j\neq i} S_j,\times_{j\neq i} A_j,\gamma_{-i})$.
 The utility function 
$u_{i,\cG} :\times_{i=1}^n S_i \ra \bR$ is then defined as follows:
for any $s\in \times_{i=1}^n S_i$, 
$u_{i,\cG}(s_i,s_{-i})\ud u_i(s_i,\rho_{-i}(s_{-i}))$.    

\begin{lemma}\label{lemma-sm} \ \ 
If $u_i(s_i, s_{-i})$ has increasing differences (the single crossing property) then 
$u_{i,\cG}(s_i, s_{-i})$ has increasing differences (the single crossing property). 
Also, if $u_i(s_i, \cdot)$ is monotone then $u_{i,\cG}(s_i, \cdot)$ is monotone. 
\end{lemma}
\begin{proof}
Assume that $(s_i,s_{-i}) \leq (t_i,t_{-i})$. Hence, $s_{-i}\leq_{-i} t_{-i}$, so that, 
by monotonicity of $\rho_{-i}$, $\rho_{-i}(s_{-i}) \leq_{-i} \rho_{-i}(t_{-i})$, and, in turn, 
$(s_i,\rho_{-i}(s_{-i})) \leq (t_i,\rho_{-i}(t_{-i}))$.
Then:
\begin{align*}
u_{i,\cG}(t_i, s_{-i}) - u_{i,\cG}(s_i, s_{-i}) & = \quad\text{[by definition]}\\ 
u_i(t_i, \rho_{-i}(s_{-i})) - u_i(s_i,\rho_{-i}(s_{-i})) &\leq \quad\text{[since $u_i$ has increasing differences]}\\
u_i(t_i, \rho_{-i}(t_{-i})) - u_i(s_i, \rho_{-i}(t_{-i}))& = \quad\text{[by definition]}\\ 
u_{i,\cG}(t_i, t_{-i}) - u_{i,\cG}(s_i, t_{-i}).&
\end{align*}
The single crossing property for $u_{i,\cG}(s_i, s_{-i})$ can be proved similarly. Let $s_{-i} \leq_{-i} t_{-i}$, so that, 
by monotonicity of $\rho_{-i}$,
$\rho_{-i}(s_{-i}) \leq_{-i} \rho_{-i}(t_{-i})$. 
Then, by monotonicity of $u_i(s_i, \cdot)$, we obtain:
$u_{i,\cG}(s_i, s_{-i}) = u_i(s_i, \rho_{-i}(s_{-i})) = u_i(s_i, \rho_{-i}(t_{-i})) = u_{i,\cG}(s_i, t_{-i})$, 
thus proving the monotonicity of $u_{i,\cG}(s_i, \cdot)$.
\end{proof}

Moreover, let us point out that if $u_i(\cdot, s_{-i})$ is (quasi)supermodular then, obviously,  $u_{i,\cG}(\cdot, s_{-i})$ remains
(quasi)super\-modular as well, so that by defining the game $\Gamma_\cG \ud \langle S_i,u_{i,\cG}\rangle_{i=1}^n$
we obtain the following consequence. 

\begin{corollary}\label{coro-sm3}
If $\Gamma$ is (quasi)supermodular 
then $\Gamma_\cG$ is (quasi)supermodular.
\end{corollary}

We call $\Gamma_\cG$ a \emph{game with abstract best response} because the $i$-th best response correspondence
$B_{i,\cG}:S_{-i} \ra \SL(S_i)$ is such that $B_{i,\cG}(s_{-i}) = \{s_i\in S_i~|~ \forall x_i\in S_i. u_i(x_i,\rho_{-i}(s_{-i}))\leq 
u_i(x_i,\rho_{-i}(s_{-i}))\}=B_{i}(\rho_{-i}(s_{-i}))$, so that the best response correspondence satisfies $B_\cG(s) = B_\cG(\rho(s))=B(\rho(s))$,
namely, $B_\cG$ can be viewed as the restriction of $B$ to the abstract strategy space $\rho(S)$. 

\begin{corollary}[\textbf{Correctness of Games with Abstract Best Response}]\label{coro-abs-resp}
Let $\cG=(\alpha_i,S_i,A_i,\gamma_i)_{i=1}^n$ be a family of 
GCs. Then,  $\Eq(\Gamma) \preceq_{\EM} \Eq(\Gamma_\cG)$ and,
in particular, $\lne(\Gamma) \leq  \lne(\Gamma_\cG)$ and $\gne(\Gamma) \leq \gne(\Gamma_\cG)$.
\end{corollary}
\begin{proof}
Since,  by Corollary~\ref{coro-sm3}, $\Gamma_\cG$ is (quasi)supermodular, we have that 
$\Eq(\Gamma)=\Fix(B)$ and
$\Eq(\Gamma_\cG)=\Fix(B_\cG)$. We have that for any $s\in \times_{i=1}^n S_i$, by extensiveness of $\rho$, $s\leq \rho(s)$, 
so that, since $B$ is monotone, we obtain $B(s) \preceq_{\EM} B(\rho(s)) =B_\cG(s)$. Hence, by Corollary~\ref{coro4}~(4), 
we obtain that $\Fix(B) \preceq_{\EM} \Fix(B_\cG)$.
\end{proof}

\begin{example}\rm 
Let us consider the two-player game $\Gamma=\langle S_i,u_i\rangle_{i=1}^2$ in \cite[Example~8.53]{ch11}, which is
a further example of Bertrand oligopoly, where 
$S_1=S_2 = [\frac{3}{2},\frac{5}{2}]\times [\frac{3}{2},\frac{5}{2}]$ and the utility functions $u_i:S_1\times S_2 \ra \bR^2$ are
defined by $u_i((s_{i1},s_{i2}),s_{-i})=(u_{i1}(s_{i1},s_{-i}),u_{i2}(s_{i2},s_{-i}))$ $\in \bR^2$ with
\begin{align*}
u_{11}(s_{11},s_{21},s_{22}) &\ud \big(52-21s_{11}+s_{21}+4s_{22}+8\sgn(s_{21}s_{22}-4 )\big)(s_{11}-1)\\
u_{12}(s_{12},s_{21},s_{22}) &\ud \big(51-21s_{12}-\sgn(s_{12}-\frac{11}{5}) +2s_{21}+3s_{22} +4\sgn(s_{21}+s_{22}-4) \big)(s_{12}-\frac{11}{10})\\
u_{21}(s_{21},s_{11},s_{12}) &\ud \big(50-20s_{21}-\sgn(s_{21}-\frac{11}{5}) +3s_{11}+2s_{12} +2\sgn(s_{11}+s_{12}-4) \big)(s_{21}-\frac{11}{10})\\
u_{22}(s_{22},s_{11},s_{12}) &\ud \big(49-20s_{22}+4s_{11}+s_{12}+\sgn(s_{11}s_{12}-4 )\big)(s_{22}-1)
\end{align*} 
Since any utility function $u_{ij}(s_{ij},s_{-i})$ does not depend on $s_{i,-j}$, let us observe that $u_i(\cdot,s_{-i}):S_i \ra \bR^2$ is supermodular. 
Moreover, by \cite[Propositions~8.56, 8.57]{ch11}, we also have that $u_i(s_1,s_2)$ has the single crossing property,
so that $\Gamma$  is indeed quasisupermodular. 
Also,
since $S_i$ is a compact (for the standard topology)
complete sublattice of $\bR^2$, we also have that $u_i(\cdot,s_{-i})$ is order upper semicontinuous, so that, for any $s\in S_1\times S_2$,
the best response correspondence $B$ satisfies $B(s)\in \SL(S_1 \times S_2)$.  
Indeed, as observed in  \cite[Example~8.53]{ch11}, 
it turns out that the utility functions $u_{ij}(\cdot,s_{-i}):[\frac{3}{2},\frac{5}{2}]\ra \bR$ have unique maximum points denoted
by $f_{ij}(s_{-i})$ which are the solutions of the equations $\frac{d}{ds}u_{ij}(s,s_{-i})=0$. An easy computation then provides:
\begin{align*}
f_{11}(s_{21},s_{22}) &\ud \frac{73}{42} +\frac{1}{42}s_{21} +\frac{2}{21}s_{22} +\frac{4}{21}\sgn(s_{21}s_{22}-4 )\\
f_{12}(s_{21},s_{22}) &\ud \frac{247}{140} +\frac{1}{42}s_{21} +\frac{1}{14}s_{22} +\frac{2}{21}\sgn(s_{21}+s_{22}-4 )\\
f_{21}(s_{11},s_{12}) &\ud \frac{9}{5} +\frac{3}{40}s_{11} +\frac{1}{20}s_{12} +\frac{1}{20}\sgn(s_{11}+s_{12}-4 )\\
f_{22}(s_{11},s_{12}) &\ud \frac{69}{40} +\frac{1}{10}s_{11} +\frac{1}{40}s_{12} +\frac{1}{40}\sgn(s_{11}s_{12}-4 )
\end{align*} 
so that the best response $B$ can be simplified as follows:
$$B(s_{11},s_{12},s_{21},s_{22}) = \big\{\big(f_{11}(s_{21},s_{22}),f_{12}(s_{21},s_{22}),f_{21}(s_{11},s_{12}),f_{22}(s_{11},s_{12})\big)\big\}.$$
As shown in \cite[Example~8.53]{ch11}, direct solutions of $\Gamma$ can be obtained by solving a linear system of four equations with four real
variables and this yields the following least and greatest equilibria: 
\begin{align*}
\lne(\Gamma)&=\Big(\frac{4940854}{2778745},\frac{5281784}{2778745},\frac{5497457}{2778745},\frac{10699993}{5557490}\Big)\\
\gne(\Gamma)&=\Big(\frac{6033654}{2778745},\frac{5848294}{2778745},\frac{5885617}{2778745},\frac{11224753}{5557490}\Big)
\end{align*}

\noindent
Carl and Heikkil{\"{a}} \cite[Example~8.58]{ch11} describe how to derive algorithmically approximate solutions of $\Gamma$ by 
approximating the fractional part of real numbers through the floor function, namely, the greatest rational number with 
$N$ fractional digits which is not more than a given real number. In this section we gave an 
abstract interpretation-based methodology for systematically designing this kind of 
approximate solutions which generalizes the approach in \cite[Example~8.58]{ch11}.
Here, we use the ceil abstraction of real numbers already described  in 
Example~\ref{ex-ceil}. Thus, we consider the closure operator $\cl_3:[\frac{3}{2},\frac{5}{2}]\ra [\frac{3}{2},\frac{5}{2}]$, 
that is, $\cl_3(x)$ is the smallest rational number with at most 3 fractional digits not less than $x$.  With a slight abuse
of notation, $\cl_3$ is also used to denote the corresponding 
componentwise function $\cl_3:[\frac{3}{2},\frac{5}{2}]^2\ra [\frac{3}{2},\frac{5}{2}]^2$, namely, 
$\cl_3(s_{i1},s_{i2}) = (\cl_3(s_{i1}),\cl_3(s_{i1}))$. 
Let $A_{\cl_3}\ud \{\frac{y}{10^3} \in \bQ ~|~ y \in [1500,2500]_{\bZ}\}=\{\cl_3(x)~|~x\in [\frac{3}{2},\frac{5}{2}]\}$ (and this
is a finite domain) and $A\ud
A_{\cl_3}\times A_{\cl_3}$. Then,
$(\cl_3,[\frac{3}{2},\frac{5}{2}],A_{\cl_3},\id)$ is a GC, so that, by Lemma~\ref{prod2}, $\cG_3 = (\cl_3, 
S_i, A, \id )_{i=1}^2$ is a pair of GCs. 
Let us denote by $\Gamma_{\cG_3}$ the corresponding game with abstract best response defined in Corollary~\ref{coro-sm3}, so that
$u_{i,\cG_3}(s_i,s_{-i}) = u_i(s_i,\cl_3(s_{-i}))$. Thus, it turns out that the abstract best response correspondence $B_{\cG_3}$ is
defined as follows: 
$$B(s_1,s_2) = \big\{\big(f_{11}(\cl_3(s_{2})),f_{12}(\cl_3(s_{2})),f_{21}(\cl_3(s_{1})),f_{22}(\cl_3(s_{1}))\big)\big\}$$
so that, $B_{\cG_3}$ can be restricted to the finite domain $A\times A$ and therefore has a finite range. This allows us to compute 
the least and greatest equilibria of $\Gamma_{\cG_3}$ by the standard RT algorithm in Figure~\ref{fig-algo}. Through a simple C++ program, we obtain the following solutions:
\begin{align*}
\lne(\Gamma_{\cG_3})&=\Big(\frac{10669}{6000},\frac{6653}{3500},\frac{79139}{40000},\frac{77017}{40000}\Big)\\
\gne(\Gamma_{\cG_3})&=\Big(\frac{91199}{42000},\frac{14733}{7000},\frac{42363}{20000},\frac{80793}{40000}\Big)
\end{align*}
By Corollary~\ref{coro-abs-resp}, we know that these are correct approximations, i.e., $\lne(\Gamma)\leq \lne(\Gamma_{\cG_3})$ and 
$\gne(\Gamma)\leq \gne(\Gamma_{\cG_3})$. 
Both fixed point calculations $\lne(\Gamma_{\cG_3})$ and $\gne(\Gamma_{\cG_3})$ 
need 16 calls to the abstract functions $f_{ij}(a_{-i})$, for some $a_{-i}\in A_{-i}$, which provide the unique
maximum points for $u_{ij}(\cdot, a_{-i})$. It is worth noting that, even with the precision of 3 fractional digits of $\cl_3$, the maximum approximation
for these abstract solutions turns out to be $\lne(\Gamma_{\cG_3})_{22}- \lne(\Gamma)_{22}=\frac{2148733}{22229960000}=0.00009665932822$.
\qed 
\end{example}

\section{Further Work}
We investigated how the abstract interpretation technique, which is widely used for static program analysis, 
can be applied to define and calculate approximate Nash equilibria
of supermodular games, thus showing how a notion of approximation of equilibria
can be modeled by an ordering relation analogously to what happens in the standard approaches to static 
analysis of the run-time behaviors of programs. To our knowledge, this is the first contribution towards
the goal of approximating solutions of supermodular games by relying on a lattice-theoretical approach. 
We see a number of interesting avenues for further work on this subject. First, our notion of correct approximation
of a multivalued function relies on a naive pointwise lifting of an abstract domain, as specified by a Galois connection, to
Smyth, Hoare, Egli-Milner and Veinott preorder relations on the powerset, which is the range of best response 
correspondences in supermodular games. It is worth investigating whether 
abstract domains can be lifted in different and more sophisticated ways to this class of preordered powersets, in particular
by taking into account that, for a certain class of complete lattices, the Veinott ordering gives rise to complete lattices~\cite{ran15}.
Secondly, it could be interesting to investigate some further conditions which can guarantee the correctness of games with abstract
strategy spaces (cf.\ Theorem~\ref{theo-con}). The goal here is that of devising a notion of simulation between games 
whose strategy spaces are related by some form of abstraction, in order to prove that if $\Gamma'$ simulates $\Gamma$ then 
the equilibria of $\Gamma$ are approximated by the equilibria of $\Gamma'$. 
Finally, while this paper set up the abstraction framework by using very simple abstract domains, 
the general task of designing useful and expressive abstract domains, possibly endowed with widening operators for
efficient fixed point computations, 
for specific
classes of supermodular games is left as an open issue. 

\paragraph*{Acknowledgements.}
The author has been partially supported by the Microsoft Research Software Engineering 
Innovation Foundation 2013 Award (SEIF 2013) and by 
the University of Padova under the 2014 PRAT project ``ANCORE''.

\end{document}